\documentclass[sigconf]{acmart}

\usepackage{hyperref}
\hypersetup{
    colorlinks=true,
    linkcolor=black,
    urlcolor=black,
    citecolor=black,
    breaklinks=true 
}

\def\equationautorefname~#1\null{Equation~(#1)\null}
\usepackage{breakurl}
\usepackage{color}
\usepackage{bm}
\usepackage{xspace}
\usepackage{algorithmic}
\usepackage{algorithm}
\usepackage{mdwlist}    
\usepackage{paralist} 
\usepackage{enumitem}  

\usepackage{utfsym}

\usepackage{adjustbox}
\usepackage{array}
\usepackage{booktabs}
\usepackage{multirow}

\usepackage{array,graphicx}
\usepackage{booktabs}
\usepackage{pifont}
\usepackage{color}

\usepackage{colortbl}
\definecolor{lightgray}{gray}{0.85}
\usepackage{multirow} 
\usepackage{comment}

\usepackage{balance} 

\usepackage{fancybox} 
\usepackage{amsmath}
\usepackage{amsfonts}

\usepackage{url}
\newcommand{\propose}{{\textsc{D-Tracker}}\xspace}
\newcommand{\ModelEstimation}{{\textsc{ModelEstimation}}\xspace}
\newcommand{\ModelUpdate}{{\textsc{ModelUpdate}}\xspace}
\newcommand{\RankUpdate}{{\textsc{RankUpdate}}\xspace}

\newcommand{\Xd}{\hat{\mathcal{X}}_{d}}
\newcommand{\Xs}{\hat{\mathcal{X}}_{s}}
\newcommand{\Xo}{\hat{\mathcal{X}}_{o}}
\newcommand{\X}{\mathcal{X}}
\newcommand{\Xc}{\mathcal{X}^{c}}
\newcommand{\Xf}{\mathcal{X}^{f}}
\newcommand{\Xhat}{\hat{\mathcal{X}}}

\newcommand{\Xdmode}{\mathbf{X}_{d}^{\rm{(mode)}}}
\newcommand{\Xsmode}{\mathbf{X}_{s}^{\rm{(mode)}}}

\newcommand{\Gmode}{\mathbf{G}^{\rm{(mode)}}}

\newcommand{\GmodeT}{{\mathbf{G}^{\rm{(mode)}}}^{T}}
\newcommand{\Umode}{\mathbf{U}^{\rm{(mode)}}}
\newcommand{\Vmode}{\mathbf{V}^{\rm{(mode)}}}

\newcommand{\Wcore}{\mathcal{W}^{\rm{(core)}}}
\newcommand{\Wkey}{\mathbf{W}^{\rm{(key)}}}
\newcommand{\Wloc}{\mathbf{W}^{\rm{(loc)}}}
\newcommand{\Wmode}{\mathbf{W}^{\rm{(mode)}}}
\newcommand{\Wmoden}{\mathbf{W}^{\rm{(\overline{mode})}}}
\newcommand{\numerator}{\mathbf{P}^{\rm{(mode)}}}
\newcommand{\denominator}{\mathbf{Q}^{\rm{(mode)}}}
\newcommand{\unfold}[2]{\rm{unfold}(#1, #2)}

\newcommand{\Stime}{\mathbf{S}^{\rm{(time)}}}
\newcommand{\Skey}{\mathbf{S}^{\rm{(key)}}}
\newcommand{\Sloc}{\mathbf{S}^{\rm{(loc)}}}
\newcommand{\Smode}{\mathbf{S}^{\rm{(mode)}}}
\newcommand{\Si}{\mathbf{S}^{(i)}}

\newcommand{\F}{\mathcal{F}}

\newcommand{\lc}{L_c}
\newcommand{\lf}{L_f}

\newcommand{\argmin}{\mathop{\rm arg~min}\limits}

\newcommand{\bhline}[1]{\noalign{\hrule height #1}}

\newcommand{\costM}[1]{<#1>}
\newcommand{\costC}[2]{<#1|#2>}
\newcommand{\costT}[2]{<#1;#2>}
\newcommand{\cF}{c_F}

\newcommand{\mypara}[1]{\vspace{0em}\noindent\textbf{#1.}}

\newcommand{\device}{\textit{Device}\xspace}
\newcommand{\vod}{\textit{VoD}\xspace}
\newcommand{\pythonlib}{\textit{Pythonlib}\xspace}
\newcommand{\chatapp}{\textit{Chatapp}\xspace}

\newcommand{\programming}{\textit{Language}\xspace}
\newcommand{\covid}{\textit{Covid-19}\xspace}

\newtheorem{problem}{Problem}
\newtheorem{lemma}{\textsc{Lemma}}
\newtheorem{definition}{Definition}

\newcommand{\hide}[1]{}

\AtBeginDocument{%
  \providecommand\BibTeX{{%
    \normalfont B\kern-0.5em{\scshape i\kern-0.25em b}\kern-0.8em\TeX}}}

\copyrightyear{2025}
\acmYear{2025}
\setcopyright{acmlicensed}
\acmConference[KDD '25]{Proceedings of the 31st ACM SIGKDD Conference on Knowledge Discovery and Data Mining V.1}{August 3--7, 2025}{Toronto, ON, Canada.}
\acmBooktitle{Proceedings of the 31st ACM SIGKDD Conference on Knowledge Discovery and Data Mining V.1 (KDD '25), August 3--7, 2025, Toronto, ON, Canada}
\acmDOI{10.1145/3690624.3709192}
\acmISBN{979-8-4007-1245-6/25/08}
\begin{document}

\title{D-Tracker: Modeling Interest Diffusion \\ in Social Activity Tensor Data Streams}

\author{Shingo Higashiguchi}
\affiliation{
  \institution{SANKEN, Osaka University}
  \city{Osaka}
  \country{Japan}
}
\email{shingo88@sanken.osaka-u.ac.jp}

\author{Yasuko Matsubara}
\affiliation{%
  \institution{SANKEN, Osaka University}
  \city{Osaka}
  \country{Japan}
}
\email{yasuko@sanken.osaka-u.ac.jp}

\author{Koki Kawabata}
\affiliation{%
  \institution{SANKEN, Osaka University}
  \city{Osaka}
  \country{Japan}
}
\email{koki@sanken.osaka-u.ac.jp}

\author{Taichi Murayama}
\affiliation{%
  \institution{Yokohama National University}
  \city{Yokohama}
  \country{Japan}
}
\email{murayama-taichi-bs@ynu.ac.jp}

\author{Yasushi Sakurai}
\affiliation{%
  \institution{SANKEN, Osaka University}
  \city{Osaka}
  \country{Japan}
}
\email{yasushi@sanken.osaka-u.ac.jp}

\begin{abstract}
    \label{sec:abstract}
    Large quantities of social activity data, such as weekly web search volumes and the number of new infections with infectious diseases, reflect peoples' interests and activities.
It is important to discover temporal patterns from such data and to forecast future activities accurately.
However, modeling and forecasting social activity data streams is difficult because they are high-dimensional and composed of multiple time-varying dynamics such as trends, seasonality, and interest diffusion.
In this paper, we propose \propose, a method for continuously capturing time-varying temporal patterns within social activity tensor data streams and forecasting future activities.
Our proposed method has the following properties:
(a) \textit{Interpretable}: it incorporates the partial differential equation into a tensor decomposition framework and captures time-varying temporal patterns such as trends, seasonality, and interest diffusion between locations in an interpretable manner; 
(b) \textit{Automatic}: it has no hyperparameters and continuously models tensor data streams fully automatically; 
(c) \textit{Scalable}: the computation time of \propose is independent of the time series length.
Experiments using web search volume data obtained from GoogleTrends, and COVID-19 infection data obtained from COVID-19 Open Data Repository show that our method can achieve higher forecasting accuracy in less computation time than existing methods while extracting the interest diffusion between locations.
Our source code and datasets are available at \url{https://github.com/Higashiguchi-Shingo/D-Tracker}.
\end{abstract}

\ccsdesc[300]{Computing methodologies~Factor analysis}

\keywords{Time series, Tensor decomposition, Interest diffusion}


\maketitle

\section{Introduction}
    \label{sec:intro}

The rapid development and proliferation of various web services has made it possible to obtain a large amount of data on users' social activities.
Social activity data, including web search logs, purchase histories, social media posts, and the number of new infections with infectious diseases, serve as a rich tapestry reflecting varied user interests and activities~\cite{memetracker,attentionshift, flu}. 
Crucially, the effective modeling and forecasting of such data are pivotal for several domains: they aid our understanding of user behavior patterns~\cite{returnexplore,talentflow}, enhance recommendation~\cite{recommend,collabofilter}, and improve demand forecasting~\cite{DeepAR}.
Also, discovering hidden patterns and predicting future events from social activity data can provide important insights into people's decision-making.
For example, marketers want to know the regions where users are interested in the product, between which regions interest has diffused, and how interest in the product will change in the future in order to develop an appropriate advertising strategy.

Modeling and forecasting social activity data remain challenging because they are generally high-dimensional and composed of complex, dynamic temporal patterns.
For example, we consider web search volume data, a 3rd-order tensor stream consisting of three attributes (\textit{keyword}, \textit{location}, \textit{time}).
It has the following characteristics that make accurate modeling difficult: 
\textbf{(a)~Complex temporal patterns}: 
web search volume data exhibit both trends and seasonal patterns.
The nature of these trends varies depending on the keywords and countries analyzed. 
Seasonal patterns, often influenced by annual discount sales and new product launches, also differ in their intensity and timing across different keywords and countries.
In addition, web search volume data exhibit interest diffusion patterns \cite{diffusion_survey}. 
Specifically, the popularity of a keyword in one location can diffuse and affect search volumes in other locations. 
If a keyword gains interest in one country and soon after in another, it can be assumed that the interest has diffused.
Moreover, data streams may contain sudden outliers.
\textbf{(b)~Dynamic changes in those patterns over time}: 
these trends, seasonality, and diffusion patterns change dynamically over time. 
Specifically, trends may shift regularly, and the diffusion of interest is often a temporary phenomenon; for example, it may be limited to a short period immediately after a product launch. 
This dynamic shift in multiple patterns over time makes modeling with a single model difficult, as models learned from historical data may no longer be useful for predicting current or future values \cite{viz_conceptdrift, conceptdrift}.
This motivates the incremental updating of models to adapt them to the latest data.
\textbf{(c)~High-dimensional tensor streams}:
the web search volume data constitute a high-dimensional tensor stream consisting of dozens or even over one hundred countries and multiple keywords, with no time mode boundaries (i.e., semi-infinite length).
This not only increases computation time but also makes the interpretation of modeling results difficult.

In this paper, we propose \propose , designed to capture trends, seasonal patterns, and interest diffusion patterns in temporal tensor streams.
Our model incorporates the reaction-diffusion equation \cite{holmes1994partial} into a tensor decomposition framework \cite{kolda}, which represents high-dimensional tensor streams in the form of latent dynamics and factors for their multilinear summation.
This allows the extraction of interest diffusion patterns in an interpretable manner while addressing high-dimensional problems.
Furthermore, the proposed algorithm can incrementally update the parameter and adaptively detect shifting points in patterns, which allows the continuous modeling of a semi-infinite collection of tensor streams.

\subsection{Preview of our results}

\begin{figure}[t]
    \centering
    \includegraphics[width=0.99\linewidth]{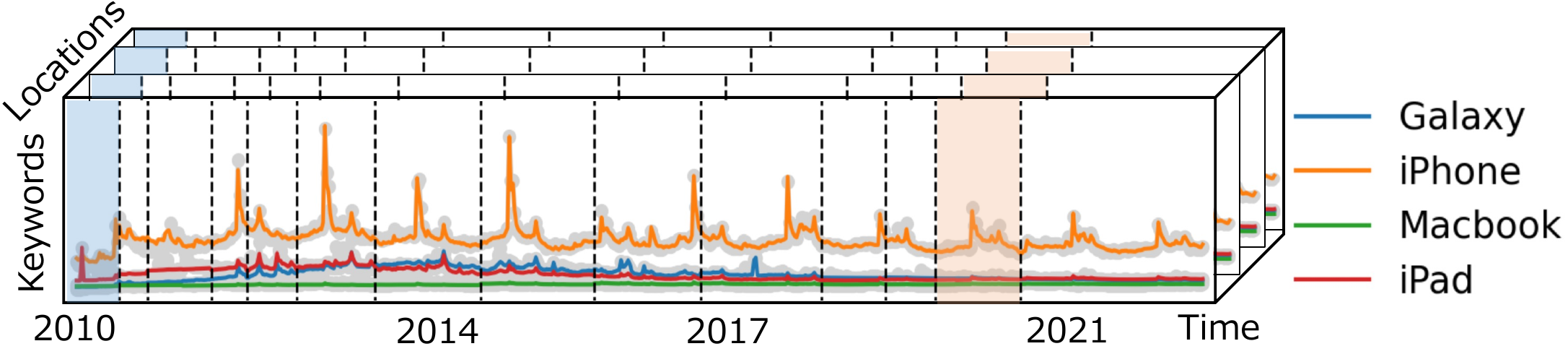}\\
    (a) Original data stream (gray lines) vs. fitting result (colored lines)\\
    
    \includegraphics[width=0.99\linewidth]{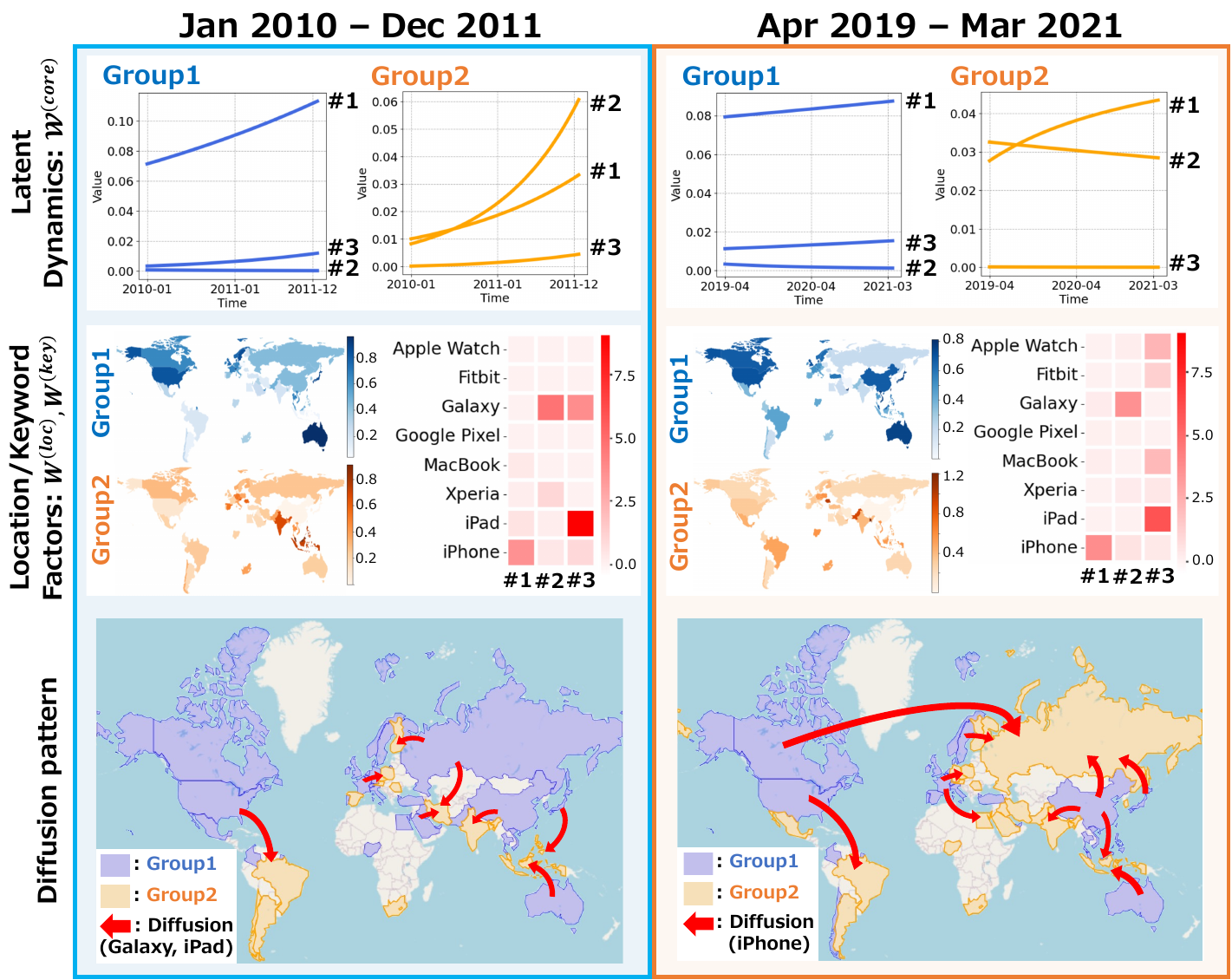} \\
    (b) Latent dynamics, location/keyword factor, and interest diffusion change dynamically
    
    \caption{Modeling power of \propose for a tensor stream related to devices: 
    (a) It fits the original tensor stream very well
    owing to model switching (at time points shown by vertical dotted lines)
    for continuous modeling and future forecasting.
    (b) It decomposes the most recent tensor
    into the latent dynamics $\Wcore$ (top) and
    the location/keyword factors $\Wloc$/$\Wkey$ (middle).
    More importantly,
    we can easily interpret patterns of interest diffusion over locations
    from the obtained model to generate $\Wcore$ (bottom).
    }
    \label{fig:figure1}
\end{figure}

\autoref{fig:figure1} shows the results of \propose for the Device-related data, which consists of weekly web search volume for eight keywords in 50 countries.
\autoref{fig:figure1} (a) shows the fitting result.
\propose can capture trends differing in terms of keywords and seasonal patterns, such as the growing interest in ``iPhone'' resulting from the release of a new model each fall.

\autoref{fig:figure1} (b) shows the modeling results for two periods (2010-2011 and 2019-2021).
Our model factorizes the input tensor into the latent dynamics $\Wcore$ (top part of \autoref{fig:figure1} (b)) and the location/keyword factors $\Wloc$, $\Wkey$ (middle part).
The latent dynamics $\Wcore$ is generated by the reaction-diffusion system, which can represent trends and inter-regional diffusion patterns, and users can easily interpret those patterns from the system's parameters.
This latent dynamics is a smaller sized tensor than the input tensor, and the input tensor is reconstructed by a multilinear summation of $\Wcore$, $\Wloc$, and $\Wkey$ (i.e., the reconstructed tensor $\Xd = \Sigma^{3}_{i=1}\Sigma^{2}_{j=1}\Wcore_{i,j} \circ \Wkey_{i}\circ \Wloc_{j}$).
The location/keyword factors are non-negative matrices, indicating the degree to which each dynamics is related to each country/keyword (A higher value indicates a stronger association between the location/keyword and the corresponding dynamics).
Note that here the latent dynamics are $(3\times 2)$-dimensional, and they are sliced and illustrated as two groups with three dynamics.
The exact definitions of $\Wcore$, $\Wloc$, $\Wkey$, and our reaction-diffusion system are given in Section \ref{sec:model}.

From the connection between the reaction-diffusion system and the two factors, we can interpret the  circumstances of web activities.
For example, in 2010-2011, Dynamics $\# 2$ of Group 2 shows a rapidly increasing trend.
The keyword factor ($\# 2$) shows that Dynamics $\# 2$ is strongly related to ``Galaxy'', and the location factor (Group 2) shows that this dynamics is strongly related to South-East Asia and South America.
These results indicate that the interest in ``Galaxy'' has an increasing pattern in these areas.
In addition, our reaction-diffusion system can describe diffusion patterns between locations.
In 2010-2011, our reaction-diffusion system showed that for Dynamics $\# 3$ there is a pattern of diffusion from Group 1 to Group 2.
This indicates that interest in ``Galaxy'' and ``iPad'', which are related to Dynamics $\# 3$, diffused from countries strongly related to Group 1 to those strongly related to Group 2. 
The bottom part of \autoref{fig:figure1} (b) shows the diffusion pattern on a map.
Furthermore, our algorithm can adapt to dynamic changes in patterns by incrementally updating the model.
Specifically, it automatically detects changes in patterns and switches the current model to a new model.
The vertical dotted lines in \autoref{fig:figure1} (a) indicate the times at which the model was switched.
By using different models in each block, \propose is able to continuously capture time-varying trends and interest diffusion in the data stream.
The right-hand side of \autoref{fig:figure1} (b) shows the modeling results for 2019-2021.
Our model captured the diffusion from Group 1 to Group 2 for Dynamics $\# 1$.
This indicates that interest in ``iPhone'' (the keyword strongly related to Dynamics $\# 1$) diffused from the US, Europe and China (Group 1) to the rest of the world (Group 2).


The main characteristics of \propose are as follows:
\begin{itemize}[leftmargin=10pt,nosep]
    \item \textbf{Interpretable}: Our proposed model incorporates the reaction-diffusion system into a tensor decomposition framework, which captures important trends and diffusion patterns in an interpretable manner
    (i.e., \autoref{fig:figure1}(b)).
    \item \textbf{Automatic}:
    The number of dimensions of the latent dynamics
    (e.g., two location groups
    and three keyword groups in \autoref{fig:figure1}(b))
    and the timing of model switching
    (e.g., the vertical dotted lines in \autoref{fig:figure1}(a)) is determined fully automatically with an objective function based on a lossless encoding scheme, which makes it easy to analyze large tensor streams.
    \item \textbf{Scalable:}
    The computation time of \propose is independent of the time series length thanks to incremental model updating and thus is a faster algorithm
    than its competitors
    (See \autoref{fig:scalability}).
    Despite the quickness (up to 8.0x faster), moreover,
    our algorithm outperforms
    the state-of-the-art forecasting approaches
    (up to 57\% improvement)
    in six real-world tensor streams.
\end{itemize}

\section{Related Work}
    \label{sec:related}

\begin{table}[t]
    \centering
    \footnotesize
    \caption{Capabilities of approaches.}
    \begin{tabular}{l|cc|cccc|c}
    \toprule
    \textbf{} & \rotatebox{90}{DISMO \cite{dismo}} & \rotatebox{90}{FluxCube \cite{fluxcube}} & \rotatebox{90}{PatchTST \cite{PatchTST}} & \rotatebox{90}{Autoformer \cite{Autoformer}} & \rotatebox{90}{DeepAR \cite{DeepAR}} & \rotatebox{90}{LaST/CoST \cite{LaST,CoST}} & \rotatebox{90}{\propose}  \\
    \midrule
    Forecasting & $\checkmark$ & $\checkmark$ & $\checkmark$ & $\checkmark$ & $\checkmark$ & $\checkmark$ &  $\checkmark$ \\
    Seasonality & $\checkmark$ & $\checkmark$ & & $\checkmark$ & & $\checkmark$ & $\checkmark$ \\
    Modeling of diffusion & & $\checkmark$ & &  &  & & $\checkmark$ \\
    Independent of data length & & & &              &              & & $\checkmark$ \\
    Parameter free    & $\checkmark$& & & & & & $\checkmark$ \\
    Outlier &  &  & & & & & $\checkmark$ \\
    \bottomrule
    \end{tabular}
    \label{tab:capability}
\end{table}

In this section, we briefly describe investigations related to this research.
\autoref{tab:capability} summarizes the relative advantages of \propose with regards to six aspects.
Only our method satisfies all the requirements for modeling social activity data streams.
We separate the details of previous studies into three categories: tensor analysis, time series forecasting, and social activity data modeling.

\mypara{Tensor analysis}
Many real-world data, such as web data, infectious disease data, and stock price data, are represented as temporal dense tensors.
Widely used tensor decomposition methods include Tucker~\cite{kolda}, PARAFAC~\cite{parafac}, and PARAFAC2~\cite{parafac2}, which have been extended to online algorithms~\cite{streamTucker,efficientNTD,onlineCP,spade,dpar2,dash} and have played an important role in temporal tensor data analysis, such as phenotype discovery~\cite{dpar2,copa,spartan}, anomaly detection~\cite{dash,zoomtucker,slicenstitch}, and clustering~\cite{faceclustering,cubescope}.
Various models have been proposed, integrating tensor decomposition with forecasting models to enable simultaneous pattern discovery and forecasting.
MLDS~\cite{MLDS} extends the Linear Dynamical System (LDS)~\cite{LDS} to model temporal tensor data through the integration of LDS and Tucker decomposition.
SMF~\cite{SMF} and SSMF~\cite{SSMF} focus on seasonality and simultaneously perform tensor decomposition and future forecasting.
However, as shown in \cite{cubecast,dismo}, these models struggle to capture nonlinear dynamics and exhibit limited predictive capability for social activity data streams, which are characterized by complex patterns that dynamically evolve over time.
By integrating PARAFAC with nonlinear differential equations, DISMO~\cite{dismo} achieves highly interpretable modeling and accurate future forecasting, capturing latent nonlinear dynamics.
However, it cannot capture inter-regional diffusion.

\mypara{Time series forecasting}
Conventional approaches to time series forecasting are based on statistical models such as autoregression (AR), Kalman filters (KF)~\cite{KF}, and their extensions \cite{MLDS,statespace,onlineARIMA,cubecast}.
In recent years, many models have been based on deep neural networks (DNN) thanks to their rapid development~\cite{surveydeep}.
In particular, models based on convolutional neural networks (CNN)~\cite{TCN,convolution} and recurrent neural networks (RNN)~\cite{DeepAR} capture higher-order temporal dependencies and achieve high forecasting accuracy.
More recently, several models that incorporate Transformer~\cite{transformer} have been proposed and have achieved high accuracy in long-term forecasting~\cite{Informer,Autoformer,Fedformer,PatchTST}.
PatchTST~\cite{PatchTST} incorporates patching and efficiently learns dependencies between sequences by decomposing a multi-channel time series into multiple single-channel time series.
Autoformer~\cite{Autoformer} captures periodic dependencies by decomposing the sequence into seasonal and trend components, subsequently computing autocorrelation to aggregate similar sub-sequences.
Furthermore, LaST~\cite{LaST} and CoST~\cite{CoST} infer a couple of representations depicting trends and seasonality of time series based on representation and contrastive learning, respectively.
However, while these DNN-based methods record high accuracy for various datasets, they still suffer as regards interpretability. 
In addition, training the model requires substantial computational resources, making it unsuitable for applications in streaming scenarios.

\mypara{Social activity data modeling}
Social activity data, such as web search logs, purchase histories, social media posts, and the number of people infected with infectious diseases, reflect users' interests and behaviors \cite{memetracker,attentionshift,interest}.
Modeling of such data has been applied in various areas such as user modeling \cite{talentflow,returnexplore}, epidemic forecasting \cite{flu}, demand forecasting \cite{DeepAR}, recommendation systems \cite{recommend,collabofilter}, and purchase prediction in online shopping \cite{eshop_clickstream}.
As models for web search volume data, those that capture the competitive relationships among keywords are proposed \cite{ecoweb,compcube}.
In addition, FluxCube~\cite{fluxcube} is based on the idea of physics-informed neural networks~\cite{physics_informed}, and models diffusion patterns in temporal tensors, but it is difficult to apply to streaming scenarios as it requires significant computation time to estimate latent regional groups.

\section{Proposed Model}
    \label{sec:model}

\begin{table}[t]
    \centering
    \caption{Symbols and definitions}
    \scalebox{0.80}{
    \begin{tabular}{l|l}
        \bhline{1.2pt}
        Symbol & Definition \\ \hline
        $k$, $l$ & Number of keywords and locations \\ 
        $t_c$ & Current time point \\ 
        $\X$ & Tensor stream, i.e., $\X \in \mathbb{N}^{k\times l\times t_c}$ \\
        $\lc$, $\lf$ & Length of current tensor and forecasting step \\
        $\Xc$,  & Current tensor, i.e., $\Xc = \X_{t_{c}-\lc : t_c}$ \\
        $d_{k},d_{l}$ & Number of latent dynamics for keywords and locations \\ 
        $d_s$ & Number of seasonal patterns \\ \hline
        $\Xd$, $\Xs$, $\Xo$ & Trend tensor, seasonal tensor, and outlier tensor \\
        $\Wcore$ & Latent dynamics generated from reaction-diffusion system \\
        $\Wkey ,\Wloc$ & Keyword and location factors \\
        $\Stime$ & Matrix including $d_s$ seasonal patterns \\
        $\Skey ,\Sloc$ & Seasonal factors for each keyword/location \\
        $\theta_{d}, \theta_{s}$ & Parameter sets for trend and seasonal tensors \\
        $\Theta$ & Single model parameter set, i.e., $\Theta = \{ \theta_{d}, \theta_{s}, \Xo \}$ \\
        $\F$ & Full parameter set, i.e., $\F = \{ \Theta_{1}, \Theta_{2}, ... , \Theta_{N}\}$ \\
        \bhline{1.2pt}
    \end{tabular}
    }
    \label{tab:definition}
\end{table} 
In this section, we describe our proposed model, \propose , for modeling and forecasting social activity data. 
First, we define the specific problem and related symbols, followed by the background of the proposed model, and finally, we describe the details of the proposed model.

\subsection{Problem definition}
\autoref{tab:definition} lists the main symbols that we use throughout this paper.
We assume that we have a semi-infinite collection of social activity data $\mathcal{X}=$ $\{ X_1, ... , X_t, ... , X_{t_c}\}$, where $t_c$ indicates the current time point and increases with every new time point.
We can obtain a new observation $X_{t_{c}+1}$ at every time point, and thus, the total size of $\mathcal{X}$ increases. 
Each entry $X_t=\{ x_{ij}(t) \}^{k,l}_{i,j=1}$ represents the total number of searches in $l$ locations for $k$ keywords. 
That is, the data observed up to $t_c$ is represented as a 3rd-order tensor, i.e., $\mathcal{X}\in\mathbb{N}^{k\times l\times t_c}$.
In an environment where new data are generated successively and the sequence length is semi-infinite, we must consider retaining only the most recent part of the data that is smaller than all of the observed data due to time and memory limitations.
Therefore, we define the input window as $\Xc =$ $\{ X_{t} \}^{t_c}_{t=t_c - \lc}$.
We formally define our problem as follows.

\begin{problem}[stream forecasting]
    \textbf{Given} the most recent $\lc$-long tensor $\Xc = \{ X_{t} \}^{t_c}_{t=t_{c} - \lc}$, where $t_c$ is the current time point and $\lc$ is the length of the current window, we want to \textbf{find} time-varying dynamics consisting of trends, seasonal patterns, and interest diffusion and \textbf{forecast} $\lf$-steps ahead values $X_{t_c + \lf} \in \mathbb{R}_{+}^{k\times l}$ continuously.
\end{problem}

For simplicity, in Sections \ref{sec:model} and \ref{sec:algo}, we consider the case of modeling web search volume data with the attributes ``keyword" and ``location", and use the terms ``keyword factor" and ``location factor", which can be generalized to other data.

\subsection{Background -- Reaction-diffusion equation}
Our model is inspired by the reaction-diffusion equation \cite{holmes1994partial}, which describes physical phenomena such as changes in chemical substances in time and space.
It is a partial differential equation composed of two terms, namely
\textit{the reaction term} and \textit{the diffusion term}.
The general form of the reaction-diffusion equation can be described with the following equations:

{\small
\begin{equation}
    \frac{\partial u_i}{\partial t} = f(u_i) + \sum_j D\cdot (u_{j}-u_{i}), \nonumber
\end{equation}}
\normalsize
where $u_i$ represents the concentration of the chemical at position $i$ and $D$ is the strength of diffusion.
The first term on the right-hand side is the reaction term, which represents the local dynamics of the chemical.
The second term is the diffusion term, which describes the diffusion of a chemical between points.
We consider applying reaction-diffusion equations to represent dynamics in social activity data streams.
In the context of the web search volume data, we assume that interest in a keyword diffuses from locations with high interest to those with low interest~\cite{diffusion_survey}, just as chemicals move from high-concentration areas to low-concentration areas~\cite{fick}.


\subsection{D-Tracker model}

Here, we present our model
for modeling interest diffusion in social activity data.
Figure \ref{fig:overview} shows an overview of our model.

\begin{figure}[!t]
    \centering
    \includegraphics[width=0.90\linewidth]{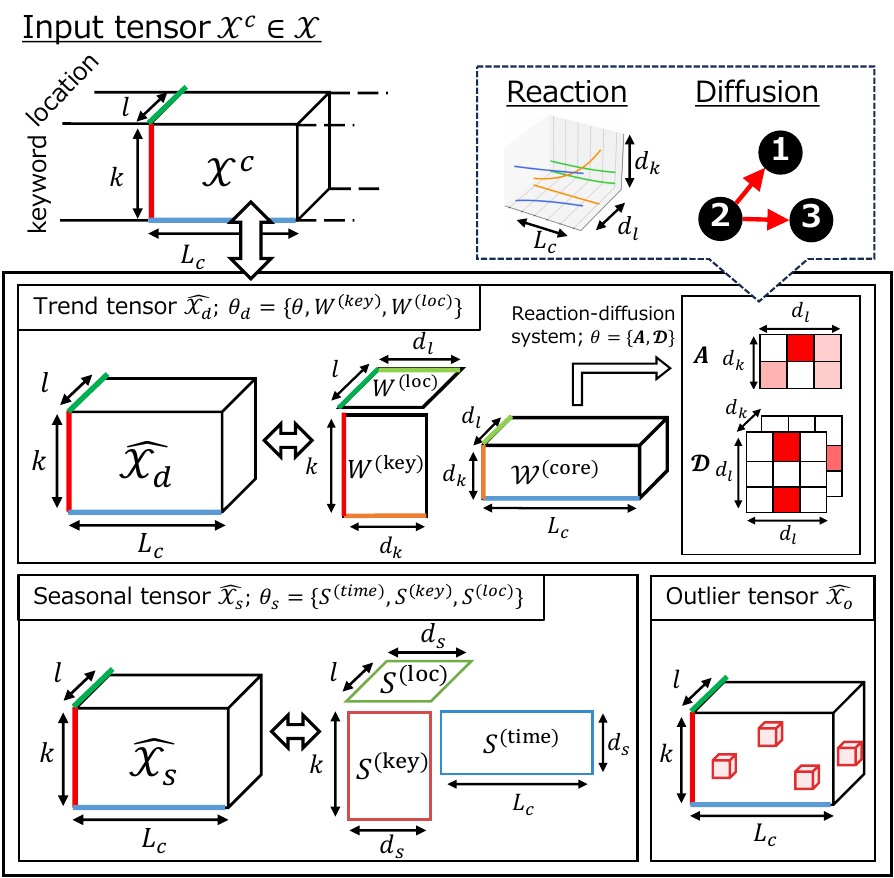}
    \caption{An overview of the \propose model: The current tensor $\Xc$ is approximated by the sum of the trend tensor $\Xd$, the seasonal tensor $\Xs$, and the outlier tensor $\Xo$. 
    }
    \label{fig:overview}
\end{figure}

The main challenge facing social activity modeling
is the accurate identification of trends and diffusion patterns
masked by multiple kinds of
seasonal patterns and sudden outliers.
To address this,
we assume the following three components:
a trend tensor $\Xd \in \mathbb{R}^{k \times l\times L_c}$,
which contains only trend and diffusion patterns,
a seasonal tensor $\Xs \in \mathbb{R}^{k \times l\times L_c}$,
which contains only seasonal patterns,
and an outlier tensor $\Xo \in \mathbb{R}^{k \times l\times L_c}$,
which contains only outliers.
The current tensor $\Xc$ is then approximated by
the sum of the components as follows.
\begin{align}
    \label{eq:ST}
    \Xc\ \approx\ \Xd\ +\ \Xs\ +\ \Xo.
\end{align}

The trend tensor $\Xd$ represents trends
for all pairs of keywords and locations
where we assume some of the locations have diffusion patterns.
Since $\Xd$ is still high-dimensional
and it is hard to interpret important trends,
we decompose the tensor into some latent factors
that can reconstruct it sufficiently.
Letting $d_k$ and $d_l$ be
the number of dimensions for such latent factors
of keywords and locations, respectively,
such that $d_k\ll k$ and $d_l\ll k$,
$\X_d$ is decomposed as: 
\begin{align}
    \label{eq:Xd}
    \Xd &= \Wcore \times_{\rm{key}} \Wkey \times_{\rm{loc}} \Wloc, 
\end{align}
where $\times_{\rm{key}}$ and $\times_{\rm{loc}}$ denote the \textit{keyword}-mode product and \textit{location}-mode product, respectively.
$\Wkey \in \mathbb{R}_{+}^{d_k \times k}$ and
$\Wloc \in \mathbb{R}_{+}^{d_l \times l}$
are non-negative matrices
used to reconstruct $\Xd$ from  
a compressed tensor $\Wcore \in \mathbb{R}^{d_k \times d_l\times L_c}$.
Crucially, $\Wcore$ describes the primary dynamics 
(i.e., trends and interest diffusion)
of the original tensor.
We thus assume that the latent dynamics $\Wcore$ is sequentially generated
by the following reaction-diffusion system.
{\small
\begin{equation}
\begin{split}
    \frac{dw_{ij}}{dt} = a_{ij}w_{ij} + \sum_{j'}d_{ijj'}(w_{ij'}-w_{ij}),
    \ (1\leq i\leq d_{k}, 1\leq j\leq d_{l}).
\end{split}
\label{eq:reactiondiffusion}
\end{equation}
}
The first term on the right-hand side is the reaction term, which expresses the increasing or decreasing trend of each of the dynamics.
The second term is the diffusion term, which expresses the pattern of diffusion among multiple dynamics.
The parameters $a_{ij}$ and $d_{ijj'}$ are interpreted as follows.
\begin{compactitem}
    \item $a_{ij}$: growth rate of dynamics $w_{ij}$; If $a_{ij}>0$, $w_{ij}$ shows an increasing trend; if $a_{ij}<0$, it shows a decreasing trend.
    \item $d_{ijj'}$: strength of diffusion between $w_{ij}$ and $w_{ij'}$; If $d_{ijj'}>0$, values of $w_{ij'}$ diffuse to $w_{ij}$.
\end{compactitem}
The entire system consists of $d_{k}\cdot d_{l}$ reaction-diffusion equations, and 
has $\mathbf{A}\in \mathbb{R}^{d_{k}\times d_{l}}$ and $\mathcal{D}\in \mathbb{R}_{+}^{d_{k}\times d_{l}\times d_{l}}$ as parameters.

\begin{definition}[Reaction-diffusion system: $\theta$]
Let $\theta = \{ \mathbf{A}, \mathcal{D}\}$ be a parameter set of our reaction-diffusion system.
Also, let $f_{\theta}(\theta ,w_0,L)$ denote the generation of an $L$-long time series, $\Wcore \in \mathbb{R}^{d_k \times d_l\times L}$,
with a given $\theta$
from the initial value $w_0\in \mathbb{R}_{+}^{d_k \times d_l}$.
\end{definition}
\noindent
Overall, the parameter set for representing the trend tensor $\Xd$
is now defined as follows.
\begin{definition}[Parameters for trend tensor: $\theta_d$]
Let $\theta_d = \{ \Wkey, \Wloc, \theta \}$ be a parameter set for a trend tensor $\Xd$.
\end{definition}

\mypara{Interpretability of $\theta_d$}
For any pair of a keyword $u\in\{1,\dots,k\}$
and a location $v\in\{1,\dots,l\}$,
the time derivative of univariate time series
$\X_{uv}\in\mathbb{R}_+^{L}$
is derived by Equations (\ref{eq:Xd}) and (\ref{eq:reactiondiffusion})
as follows.
{\small
\begin{align}
\frac{d\X_{uv}}{dt} 
=
\underbrace{\sum_{i=1}^{d_k}\sum_{j=1}^{d_l}\Wkey_{iu}\Wloc_{jv}}_{\textrm{Multi-linear projection}}
\biggl(
\underbrace{a_{ij}w_{ij}+\sum_{j'=1}^{d_l}d_{ijj'}(w_{ij'}-w_{ij})}_{\textrm{Latent reaction-diffusion system}}
\biggr). \nonumber
\end{align}
}%
The latent reaction-diffusion system describes $d_k\cdot d_l$ kinds of
important dynamics, which allow us to interpret $d_k$ 
latent keyword groups and $d_l$ 
latent location groups. 
The multi-linear projection then quantifies how strongly the latent groups contribute to the
actual search volumes for the $u$-th keyword at the $v$-th location.
Since the projection with non-negative matrices never inverts the relationships between latent dynamics, the interpretability of $\theta$ is preserved throughout the reconstruction process using $\theta_d$.

As with the trend tensor, we factorize the seasonal tensor $\Xs$ and obtain a compact representation of seasonality.
The decomposition we consider is expressed as:
\begin{align}
    \label{eq:Xs}
    \Xs =
    \Stime \times_{\rm{key}} \Skey \times_{\rm{loc}} \Sloc,
\end{align}
where $d_s$ is the number of latent seasonal patterns.
$\Stime \in \mathbb{R}^{d_s \times \lc}$ denotes the seasonal matrix, and $\Skey \in \mathbb{R}^{d_s \times k}$ and $\Sloc \in \mathbb{R}^{d_s \times l}$ denote the keyword/location factors, representing the degree to which each of the $d_s$ seasonal patterns is related to each keyword and location.

\begin{definition}[Parameter set for seasonal tensor: $\theta_s$]
Let $\theta_s = \{ \Stime, \Skey \Sloc \}$ be a parameter set for a seasonal tensor.
\end{definition}

In addition, the outlier tensor $\Xo$ stores only anomalous values.
On the assumption that such remarkable abnormal behaviors are rare, $\Xo$ should be sparse.

\begin{definition}[Single model parameter set    : $\Theta$]
Let $\Theta = \{ \theta_d , \theta_s , \Xo \}$ be a single model parameter set.
It represents trends, diffusion patterns, seasonal patterns, and outliers in the input tensor.
\end{definition}

Thus far, we have described three sets of parameters capable of capturing latent dynamics, such as trends, interest diffusion, seasonality, and outliers within the input tensor. 
However, in real-world social activity data, temporal patterns change dynamically.
The final challenge is to extend the model so that it effectively captures these time-varying patterns in tensor streams.
Our method adaptively captures time-varying patterns by switching between multiple models.
Specifically, let $N$ denote the appropriate number of patterns up to the current time point.
In this context, a tensor stream $\X$ is represented using a set of $N$ models, i.e., $\{ \Theta_{1},\Theta_{2},\dots, \Theta_{N} \}$.
The full parameter set we want to estimate is eventually defined as follows.

\begin{definition}[Full parameter set: $\F$]
Let $\F=\{\Theta_{1},\Theta_{2},
\dots,\Theta_{N} \}$ be a full parameter set of \propose .
\end{definition}

\section{Optimization Algorithms}
    \label{sec:algo}

\begin{algorithm}[t]
    \footnotesize
    \caption{\propose ($\X$)}
    \label{alg:DTracker}
\begin{algorithmic}[1]
    \REQUIRE
        Tensor stream $\X$ \\
    \ENSURE
        Full parameter set $\F$
        and predicted tensor stream $\Xhat$ \\
\textbf{Initialization:}
\STATE 
    $\{\Theta , d_{k}, d_{l}, d_{s}\}\leftarrow \rm{Initialize}(\mathcal{X}_{0:\lc})$;
    \ \ 
    $\F = \{ \Theta \}$;
\\
\textbf{Stream processing:}
\FOR{$t_c = L_{c}$ to $n$}
    \STATE $\Xc = \X_{t_{c}-L_{c}:t_{c}}$;
    \STATE $\Theta' =$ \ModelEstimation $(\Xc ,d_{k}, d_{l}, d_{s})$;
    \STATE $\{\F ,d_{k}, d_{l}, d_{s}\}=$ \ModelUpdate $(\Xc ,\F ,\Theta' ,d_{k}, d_{l}, d_{s})$;
    \STATE $\Xhat_{t_c + \lf} = \rm{Forecast}(\F)$;
\ENDFOR
\STATE {\bf return} $\F ,\Xhat$;
\end{algorithmic}
\end{algorithm}

The previous section described the mathematical concepts that enable us to obtain a reasonable representation of latent dynamics in the input tensor.
In this section, we present effective algorithms for estimating the full parameter set $\F$.

A significant challenge in parameter estimation lies in determining the optimal structure of the model (i.e., the optimal values of $d_{k}, d_{l}, d_{s}$ and the optimal number of $\Theta$ in full parameter set $\F$), which requires users to perform a difficult, time-consuming analysis.
Moreover, pre-defined models become outdated as we observe new tensors.
Therefore, we propose an algorithm to obtain an optimal summary of tensor streams while automatically determining the structure of the model without manual intervention.
The proposed algorithm is based on the minimum description length (MDL) principle and continuously updates the full parameter set while processing the current tensor $\Xc$ in a streaming fashion.
Algorithm \ref{alg:DTracker} shows the overall procedure of \propose .
Specifically, the algorithm consists of the two sub-algorithms:
(1) \ModelEstimation : given the current tensor $\Xc$ and current rank $d_{k}, d_{l}, d_{s}$, it estimates the optimal parameter $\Theta$ from scratch;
(2) \ModelUpdate : given the current tensor $\Xc$ and the newly estimated parameters $\Theta$, it updates the full parameter set $\F$.
Specifically, it determines whether or not to add $\Theta$ to $\F$ based on the total cost.
It also updates $d_{k}, d_{l}, d_{s}$.

\subsection{ModelEstimation}
\label{subsec:ModelEstimation}
Here, we describe \ModelEstimation , an optimization algorithm that estimates the optimal parameters 
$\Theta$.
To estimate the optimal parameters, we propose an optimization method based on alternating least squares (ALS).
Specifically, we alternate between updating the reaction-diffusion system and updating each factor for the trend tensor, each factor for the seasonal tensor, and the outlier tensor, while keeping the other parameters fixed.
The overall \ModelEstimation algorithm is shown as Algorithm \ref{alg:ModelEstimation} in Appendix \ref{appendix:ModelEstimation}.

\subsubsection{Update of $\theta_d$}
In this step, $\theta$, $\Wkey$ and $\Wloc$ are alternately updated to minimize the error between $\Xc - \Xs - \Xo$ and $\Xd$.
The parameters of the reaction-diffusion system, $\theta$, and its initial value are updated as follows, with the other parameters fixed:
\footnote{In the optimization, we apply the Levenberg-Marquardt (LM) algorithm~\cite{lm}, which can solve the nonlinear least squares minimization problem effectively.}
{\small
\begin{align}
\label{eq:lmfit}
    \{ \theta ,w_0 \} = \argmin_{\theta', w'_0} ||&\Xc - \Xs - \Xo \nonumber \\
    &-  f_{\theta}(\theta',w'_0,L_{c}) \times_{\rm{key}} \Wkey \times_{\rm{loc}} \Wloc||.
\end{align}
}

Based on \cite{NTD}, $\Wkey$ and $\Wloc$ can be updated as follows, with the other parameters fixed:
\small
\begin{align}
\label{eq:factorupdate}
    \numerator &= \max (\epsilon , \Xdmode \cdot \Gmode ), \nonumber \\
    \denominator &= \max (\epsilon , \Wmode \cdot \Gmode \cdot \GmodeT ), \nonumber \\
    \Wmode &= \Wmode \otimes \numerator \oslash \denominator ,
\end{align}
\normalsize
where $\cdot$, $\otimes$, and $\oslash$ are the dot product, the Hadamard product, and the element-wise division, respectively, and $\epsilon$ is a small constant.
``(mode)'' represents either ``(key)'' or ``(loc)''.
Note that $\Xdmode$ and $\Gmode$ are calculated as follows.
\small
\begin{align}
    \Xdmode &= \unfold{\Xc - \Xs - \Xo}{\rm{mode}}, \nonumber \\
    \Gmode &= \unfold{\Wcore \times_{\rm{\overline{mode}}} \Wmoden}{\rm{mode}}.
\end{align}
\normalsize
The function ``unfold'' converts a tensor into a matrix based on the specified mode.
$\overline{mode}$ represents the opposite mode to \textit{mode} (i.e., when ``(mode)'' is ``(key)'', ``($\rm{\overline{mode}}$)'' represents ``(loc)'').

\subsubsection{Update of $\theta_s$}
In this step, $\Stime$, $\Skey$ and $\Sloc$ are alternately updated to minimize the error between $\Xc - \Xd - \Xo$ and $\Xs$.
Based on \cite{kolda}, the update formula 
is expressed as follows:

\small
\begin{align}
    \Umode &= \odot_{i \neq \rm{mode}}^{M} \Si , 
    \Vmode = \otimes_{i \neq \rm{mode}}^{M} \Si{}^\mathsf{T} \Si ,\nonumber \\
    \Smode &= \Xsmode \Umode \Vmode{}^\dagger .
\end{align}
\normalsize
where $\dagger$ indicates the Moore-Penrose pseudoinverse, $M$ is the set of all modes, i.e., $M = \{  \rm{time}, \rm{key}, \rm{loc}\}$, and $\Xsmode$ is calculated as 
\small
\begin{align}
    \Xsmode &= \unfold{\Xc - \Xd - \Xo}{\rm{mode}}.
\end{align}
\normalsize

\subsubsection{Update of $\Xo$}
In this step, the outlier tensor $\Xo$ is estimated by sparsifying the residual tensor $\Xc -\Xd - \Xs$.
The criterion for sparsification is the MDL-based cost function expressed in Equation (\ref{eq:totalcost}) (cost function is detailed in Section \ref{subsec:MDL}).
We sparsify the residual tensor so that the total cost is minimized.
For example, if the element $\Xo{}_{ijk}$ is negligibly small, the cost function suggests setting its value at zero, i.e., $\Xo{}_{ijk}=0$.
By retaining only elements that reduce the total cost and setting the other elements at zero, we can obtain a sparse outlier tensor $\Xo$ that minimizes the total cost.
This allows us to estimate the outlier tensor automatically without manual thresholds.

The time complexity of \ModelEstimation is as follows.
\begin{lemma}
\label{lemma1}
    The time complexity of \ModelEstimation is
    $O(d_{k}d_{l}^{2} + kd_{k} + ld_{l})$. 
    See Appendix \ref{appendix:proof} for details.
\end{lemma}

\subsection{Automatic tensor compression}
\label{subsec:MDL}
Here, we propose a criterion for automatically choosing the best model structure of \propose , $\F$, for a given tensor stream $\X$.
We apply the minimum description length (MDL) principle, which enables us to determine the nature of a good summarization by minimizing the sum of the model description cost and the data encoding cost.
Specifically, the total cost of $\F$ for $\X$ is described as below:
\begin{align}
\label{eq:totalcost}
    \costT{\mathcal{X}}{\F} &= \costM{\F} + \costC{\mathcal{X}}{\F},
\end{align}
where $\costM{\F}$ shows the model description cost, and $\costC{\mathcal{X}}{\F}$ represents the cost of describing the data $\X$ given the model $\F$.
In short, it follows the assumption that the more we can compress the data, the more we can learn about its underlying patterns.

\mypara{Model description cost}
The model description cost $\costM{\F}$ is defined as the number of bits needed to store the model parameters.
Our model description cost is represented as shown below:
\begin{align}
    \costM{\F} = \sum\limits_{\Theta \in \F} \costM{\Theta},\ \ 
    \costM{\Theta} = \costM{\theta_d} + \costM{\theta_s} + \costM{\Xo}. \nonumber 
\end{align}
The cost of $\theta_d$ is further decomposed as follows.
\small
\begin{align}
    \label{eq:rds}
    \costM{\theta_d} &= \costM{\Wkey} + \costM{\Wloc} + \costM{\theta}, \nonumber \\
    \costM{\Wkey} &= |\Wkey|\cdot (\log d_k + \log k + \cF) + \log^{*} |\Wkey|, \nonumber \\
    \costM{\Wloc} &= |\Wloc| \cdot (\log d_l + \log l + \cF) + \log^{*} |\Wloc|, \nonumber \\
    \costM{\theta} &= \costM{\mathbf{A}} + \costM{\mathcal{D}}, \nonumber \\
    \costM{\mathbf{A}} &= |\mathbf{A}| \cdot (\log d_k + \log d_l + \cF) + \log^{*} |\mathbf{A}|, \nonumber \\
    \costM{\mathcal{D}} &= |\mathcal{D}| \cdot (\log d_k + 2\log d_l + \cF) + \log^{*} |\mathcal{D}|, \nonumber
\end{align}
\normalsize
where $|\cdot|$ shows the number of nonzero elements for a given vector/matrix, $\log^{*}$ is the universal code length for integers, and $c_F$ is the float point cost\footnote{We used $c_F = 32$ bits.}.
Similarly, the costs of $\theta_s$ and $\Xo$ are defined as follows.
\small
\begin{align}
    \costM{\theta_s} &= \costM{\Stime} + \costM{\Skey} + \costM{\Sloc}, \nonumber \\
    \costM{\Stime} &= |\Stime| \cdot (\log d_s + \log L_c + \cF) + \log^{*} |\Stime|, \nonumber \\
    \costM{\Skey} &= |\Skey| \cdot (\log d_s + \log k + \cF) + \log^{*} |\Skey|, \nonumber \\
    \costM{\Sloc} &= |\Sloc| \cdot (\log d_s + \log l + \cF) + \log^{*} |\Sloc|, \nonumber \\
    \costM{\Xo} &= |\Xo| \cdot (\log n + \log k + \log l + \cF) + \log^{*} |\Xo|. \nonumber
\end{align}
\normalsize

\mypara{Data encoding cost}
The data encoding cost measures how well the given model compresses the original data. 
The Huffman coding scheme \cite{huffman} enables us to encode $\X$ using $\F$.
It assigns a number of bits to each element in $\X$, which is the negative log-likelihood under a Gaussian distribution with mean~$\mu$ and variance~$\sigma ^{2}$, i.e.,
\begin{align}
    \costC{\mathcal{X}}{\F} = \sum_{x \in \X} -\log_2 p_{\mu, \sigma}(x-\hat{x}_d -\hat{x}_s -\hat{x}_o). \nonumber
\end{align}
where $\hat{x}_d , \hat{x}_s$ shows the reconstruction values of $x$ obtained from \autoref{eq:Xd} and \autoref{eq:Xs}.

\subsection{ModelUpdate}
\label{subsec:ModelUpdate}
Finally, the problem we want to solve is how to obtain a good summarization of the tensor stream while detecting model switching.
The overall \ModelUpdate algorithm is shown as Algorithm \ref{alg:ModelUpdate} in Appendix \ref{appendix:modelupdate}).
The main idea is to switch to the new model when it results in a reduction of the total cost of $\X$.
Given the current tensor $\Xc$, the full parameter set $\F$, and the candidate model parameter $\Theta$, the algorithm compares the total cost of adding $\Theta$ to $\F$ with the total cost without the addition. 
If the total cost is lower with the new model, we add the candidate model parameter to the full parameter set and model the current tensor using the new model.
In addition, the optimal values of $d_k$, $d_l$, and $d_s$ may change over time.
We update these values based on the total cost to better represent the current tensor.
We assume that the optimal values change gradually in streaming scenarios and only search for cases where only one of these values changes by one (See Algorithm \ref{alg:RankUpdate}: \RankUpdate in Appendix \ref{appendix:modelupdate} for details).

\section{Experimental Evaluation}
    \label{sec:experiment}

\begin{table*}[h]
    \centering
    \caption{Forecasting performance comparison
        respectively
    ($\times 10^{-2}$).
    }
    \scalebox{0.80}
    {
    \begin{tabular}{c|l|c|cc|cc|cc|cc|cc|cc|cc|cc}
        \toprule
         \multicolumn{3}{c|}{Method}& \multicolumn{2}{c|}{{\small \propose}} & \multicolumn{2}{c|}{{\small DISMO}} & \multicolumn{2}{c|}{{\small FluxCube}} & \multicolumn{2}{c|}{{\small PatchTST}} & \multicolumn{2}{c|}{{\small Autoformer}} & \multicolumn{2}{c|}{{\small DeepAR}} & \multicolumn{2}{c|}{{\small LaST}} & \multicolumn{2}{c}{{\small CoST}} \\ 
        \cmidrule(l{0mm}r{0mm}){1-3}
        \cmidrule(l{1.5mm}r{1.5mm}){4-5}
        \cmidrule(l{1.5mm}r{1.5mm}){6-7}
        \cmidrule(l{1.5mm}r{1.5mm}){8-9}
        \cmidrule(l{1.5mm}r{1.5mm}){10-11}
        \cmidrule(l{1.5mm}r{1.5mm}){12-13}
        \cmidrule(l{1.5mm}r{1.5mm}){14-15}
        \cmidrule(l{1.5mm}r{1.5mm}){16-17}
        \cmidrule(l{1.5mm}r{1.5mm}){18-19}
        ID&Dataset & $L_f$ & {\small MAE} & {\small RMSE} & {\small MAE} & {\small RMSE} & {\small MAE} & {\small RMSE} & {\small MAE} & {\small RMSE} & {\small MAE} & {\small RMSE} & {\small MAE} & {\small RMSE} & {\small MAE} & {\small RMSE} & {\small MAE} & {\small RMSE} \\ 
        \midrule
        \#1&\device &13&\textbf{0.84}&\textbf{1.82}&\underline{1.08}&2.18&3.86&7.20&1.39&3.20&1.21&\underline{2.03}&2.14&2.99&1.91&4.00&1.63&3.45 \\
              &&26&\textbf{0.90} &\textbf{1.95}&\underline{1.19}&2.45&3.85&7.26&1.54&3.19&1.31&\underline{2.13}&2.21&3.16&2.07&4.31&1.80&3.79 \\
              &&39&\textbf{0.97}&\textbf{2.18}&1.22&2.60&3.97&7.63&\underline{1.21}&2.76&1.34&\underline{2.22}&2.37&3.34&2.37&5.83&1.74&3.77 \\ 
        \midrule
        \#2&\pythonlib &13&\textbf{0.27}&\textbf{1.42}&0.47&1.74&0.86&3.57&0.89&4.26&0.49&\underline{1.45}&2.01&3.84&0.71&2.54&\underline{0.40}&2.46 \\
        &&26&\textbf{0.32}&\underline{1.54}&0.52&1.96&0.93&3.68&0.95&4.17&0.47&\textbf{1.35}&2.01&3.84&0.86&2.80&\underline{0.45}&2.76 \\
        &&39&\textbf{0.37}&\underline{1.78}&0.57&2.22&1.09&4.08&0.83&3.33&\underline{0.42}&\textbf{1.52}&2.01&3.86&0.92&2.99&0.47&3.08 \\ 
        \midrule
        \#3&\vod &13&\textbf{0.53}&\underline{1.58}&\underline{0.56}&1.64&2.83&6.96&0.95&2.17&\underline{0.56}&\textbf{1.33}&1.77&2.23&2.91&13.26&2.98&14.29 \\
        &&26&\textbf{0.59}&1.83&\underline{0.66}&1.91&2.83&6.95&0.84&\underline{1.82}&0.74&\textbf{1.63}&1.88&2.42&2.53&11.42&2.36&11.17 \\
        &&39&\textbf{0.63}&\textbf{2.02}&\underline{0.78}&2.22&2.79&6.84&1.69&2.57&1.13&\underline{2.07}&1.98&2.61&1.55&5.50&3.27&15.18 \\ 
        \midrule
        \#4&\chatapp &13&\textbf{0.28}&\underline{0.83}&\underline{0.47}&1.16&1.47&4.06&\underline{0.47}&\textbf{0.81}&0.84&1.34&1.85&2.35&2.66&13.16&1.76&7.30 \\
        &&26&\textbf{0.34}&\textbf{1.01}&\underline{0.48}&1.21&1.48&4.02&0.74&\underline{1.13}&0.81&1.43&2.16&3.04&2.34&9.91&2.09&9.11 \\
        &&39&\textbf{0.41}&\textbf{1.24}&\underline{0.50}&\underline{1.27}&1.47&3.93&0.88&1.53&0.54&1.41&1.93&2.61&2.18&8.59&2.18&7.48 \\ 
        \midrule
        \#5&\programming &13&\textbf{0.31}&\underline{0.97}&0.94&2.25&1.38&3.34&0.41&1.47&\underline{0.37}&\textbf{0.85}&1.98&3.63&0.74&1.94&0.40&1.50 \\
         &&26&\textbf{0.35}&\underline{1.14}&1.12&2.68&1.42&3.33&0.50&1.97&\underline{0.41}&\textbf{0.93}&2.03&3.65&0.80&2.13&0.47&1.73 \\
         &&39&\textbf{0.40}&\underline{1.32}&1.36&3.19&1.51&3.44&0.51&2.50&0.54&\textbf{1.08}&2.01&3.55&0.81&2.13&\underline{0.49}&1.79 \\ 
        \midrule
        \#6&\covid
        &7&\textbf{0.47}&\textbf{1.40}&1.12&2.87&1.95&2.75&1.60&5.55&\underline{1.10}&\underline{2.39}&4.60&5.41&2.39&7.91&1.23&7.07 \\
        &&14&\textbf{0.70}&\textbf{2.42}&1.33&3.94&1.97&\underline{2.85}&2.04&7.64&\underline{1.18}&2.97&2.30&3.88&2.10&6.78&1.45&8.50 \\
        &&21&\textbf{0.98}&3.65&1.58&5.83&1.98&\textbf{3.08}&2.24&17.30&\underline{1.25}&\underline{3.25}&2.61&4.40&1.74&5.71&1.72&9.14 \\
        \bottomrule
    \end{tabular}
    }
    \label{tab:accuracy}
\end{table*}

In this section, we describe the performance of \propose using real datasets. 
The experiments were designed to answer the following questions about \propose .

\begin{itemize}
\setlength{\itemindent}{-4mm}
\item [Q1.]
\textit{Accuracy:} How accurately does it forecast future values?
\item [Q2.]
\textit{Scalability:} How does it scale in terms of computational time?
\item [Q3.]
\textit{Effectiveness:} How well does it extract latent dynamic patterns?
\end{itemize}

\mypara{Datasets}
We use web search volume datasets obtained from GoogleTrends \cite{googletrend} and data on the number of COVID-19 infections obtained from the COVID-19 Open Data Repository \cite{covid-19}.
For the web search volume datasets, each tensor contains weekly web search volumes for various keywords in 48 countries from 2010 to 2022.
A wide range topic of keywords was selected to demonstrate the proposed method's ability to capture diverse interest diffusion.
The COVID-19 dataset contains the daily number of COVID-19 infections and deaths in 50 countries from January 2020 to September 2022.
All data are normalized so that the maximum value is 1 and the minimum value is 0.
The set of keywords contained by each dataset and data size are shown in Table~\ref{tab:dataset} in Appendix \ref{appendix:exp_setting}.

\mypara{Baselines}
We compare our method with the following state-of-the-art models for time series forecasting: DISMO \cite{dismo}, FluxCube \cite{fluxcube}, PatchTST \cite{PatchTST}, Autoformer \cite{Autoformer}, DeepAR \cite{DeepAR}, LaST \cite{LaST}, and CoST \cite{CoST}.
The use of these models and the experimental settings are detailed in Appendix \ref{appendix:exp_setting}.

\mypara{Q1: Accuracy}
We compared the MAE and RMSE when we varied the forecasting step $\lf$ in $\{ 13,26,39\}$-steps (i.e., 3, 6, 9 months) for GoogleTrends datasets and in $\{ 7,14,21\}$-steps (i.e., 1,2,3 weeks) for COVID-19 dataset.
For our proposed method, we set $\lc$ , the length of the input tensor, at $104$ steps (i.e., two years) for GoogleTrends datasets and $56$ steps (i.e., eight weeks) for the COVID-19 dataset.
Table \ref{tab:accuracy} shows the overall results, where the bold font and underlining show methods providing the best and second-best levels of performance, respectively.
\propose outperforms its competitors in terms of MAE on all datasets
and greatly reduces RMSE.
DISMO extracts competitive relationships
in tensor streams, resulting in relatively high accuracy
for datasets related to competitive industries
(i.e., \device, \vod, and \chatapp).
However, since it misses location dependencies
and cannot adjust the CP rank dynamically,
our method is superior to DISMO.
The offline models exhibit low forecasting accuracy
due to shifting trends
although DNN-based models can learn
complex representations of seasonalities and trends.
Note that RMSE is sensitive to large errors.
The comparable performances of our method
in terms of RMSE indicate that it produced
temporal large errors because of shifting trends
out of the current tensor
but once it observed them,
it can update the current model for subsequent forecasting
and obtained the best MAE scores.
As a result,
our method realizes accurate tensor forecasting
with being aware of diffusion patterns in social activities.

\begin{figure}[h]
    \centering
    \includegraphics[width=0.95\columnwidth]{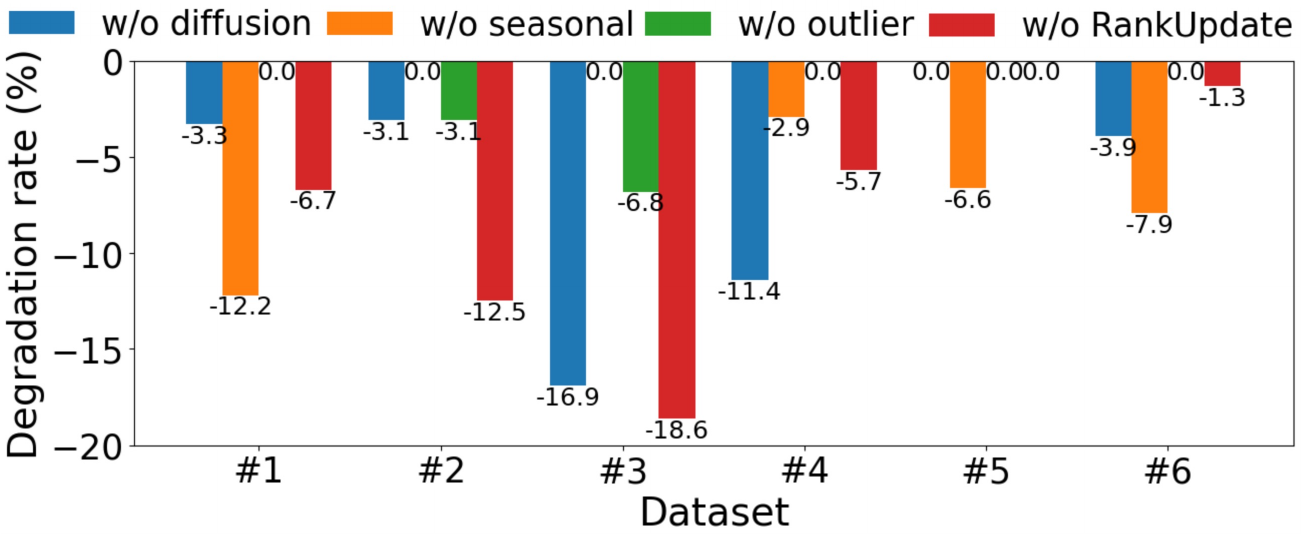}
    \caption{Ablation study results. The numbers under the bars indicate the improvement rate for each. Each mechanism in our model provides an improvement in accuracy.}
    \label{fig:ablation}

    \begin{tabular}{ll}
    \begin{minipage}[b]{0.47\linewidth}
        \centering
        \includegraphics[width=\linewidth]{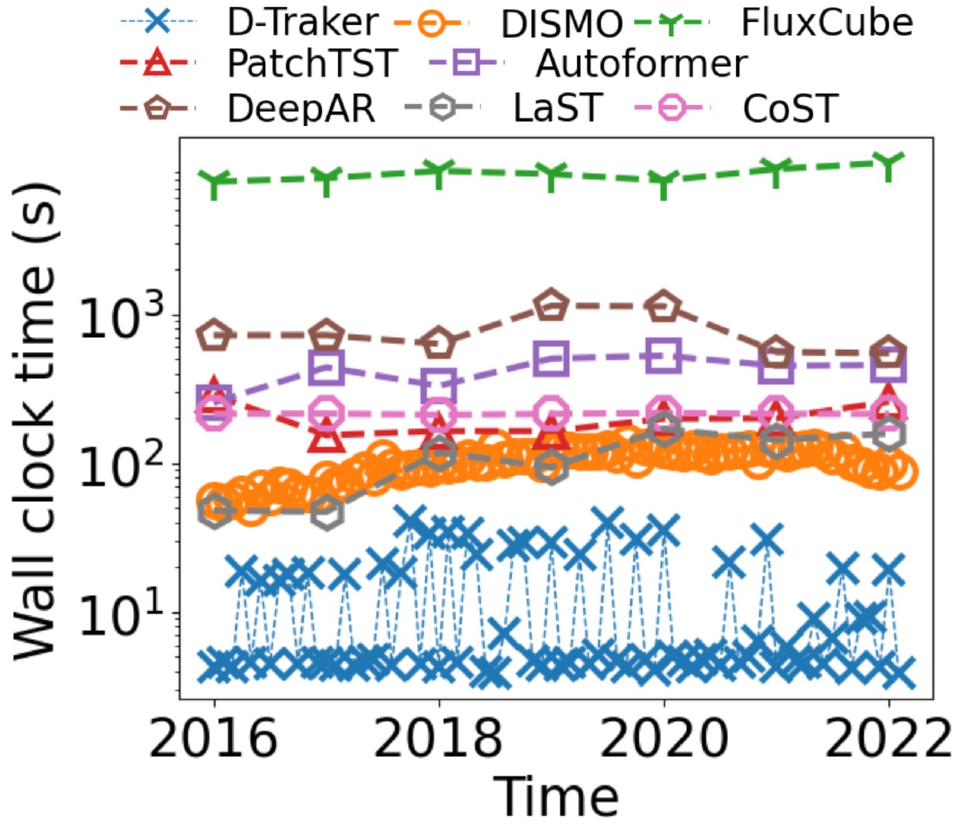}
    \end{minipage}
    &
    \begin{minipage}[b]{0.46\linewidth}
        \centering
        \includegraphics[width=\linewidth]{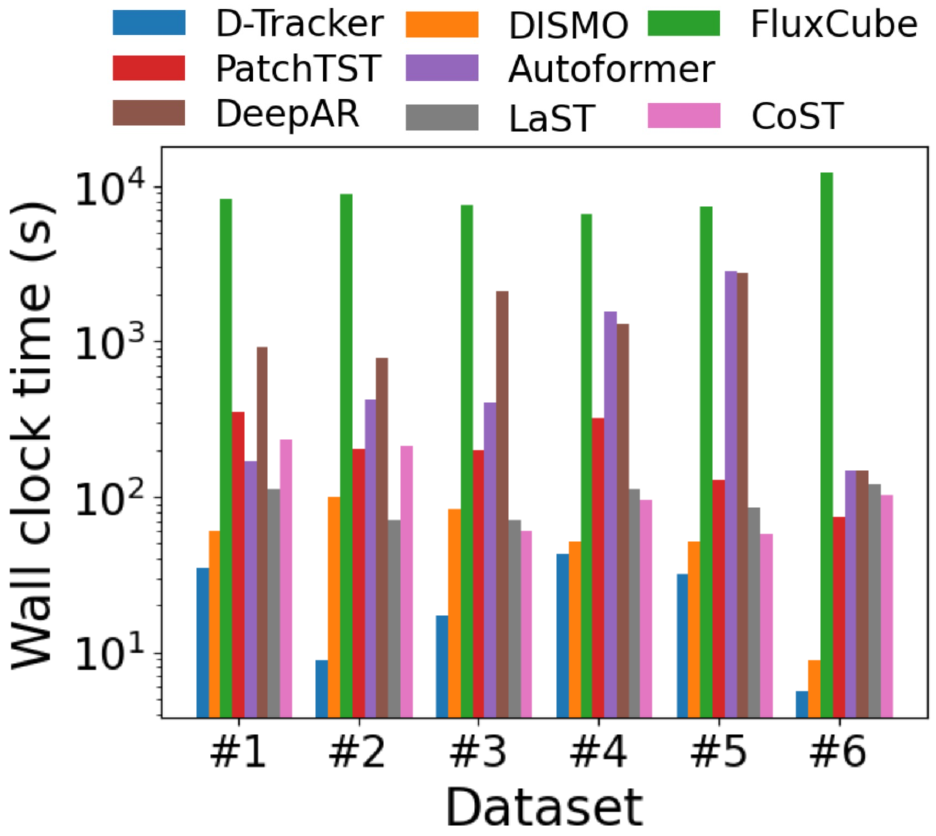}
    \end{minipage}
  \end{tabular}
  \caption{Scalability of \propose :(left) Wall clock time vs. data stream length and (right) average time consumption.}
  \label{fig:scalability}
\end{figure}

\mypara{Ablation study}
We performed ablation experiments to demonstrate the effectiveness of all the components of \propose.
We quantified the contribution of each component to the overall performance by calculating the degradation rate, defined as $\frac{\mathcal{L}_{full} - \mathcal{L}_{abl}}{\mathcal{L}_{full}}$, where $\mathcal{L}_{full}$ and $\mathcal{L}_{abl}$ represent the MAEs for the full model and the ablation models, respectively, over forecast periods of 26 weeks or 14 days.
``Full model w/o diffusion,'' ``w/o seasonal,'' and ``w/o outlier'' refer to the models where the diffusion term, seasonality modeling, and outlier capture mechanisms have been removed, respectively. 
Additionally, to evaluate the impact of the dimensionality of the latent dynamics on accuracy, we also conducted experiments on ``Full model w/o \RankUpdate ''.
\autoref{fig:ablation} shows the overall results.
In datasets with high seasonality (\device , \programming , and \covid), $\theta_s$ works to improve the accuracy.
The diffusion term significantly improved the accuracy in the \vod and \chatapp datasets.
The outlier mechanism filters out abrupt changes and helps capture trends in some datasets, thereby improving the accuracy.
The \RankUpdate algorithm enhances predictive accuracy across the entire data stream by dynamically adjusting the model’s structure to adapt to evolving patterns over time.
These results show that each component in \propose contributes to effectively modeling social activity data streams.
On the other hand, in datasets with no (or very few) seasonal patterns or outliers, the seasonal and outlier terms do not capture important patterns and do not lead to improved forecast accuracy. 
Also, if the optimal number of latent dynamics to represent the data stream did not change throughout, the ranks would not be updated and the ablation model predictions would be exactly the same as the full model predictions. 
Furthermore, in Dataset 5, the diffusion term did not improve accuracy. 
This means that the model was able to represent the data stream well without the diffusion term, suggesting that there was little interest diffusion regarding the programming languages. 
Importantly, while the data stream consists of various patterns, whose presence or absence cannot be predicted in advance, our methods automatically adapt to these patterns.

\mypara{Q2: Scalability}
We next evaluated the performance of \propose in terms of computational time.
The left column in \autoref{fig:scalability} shows the wall clock time of an experiment performed on the \pythonlib dataset.
Thanks to the incremental parameter updates, the computation time of \propose is independent of the data stream length.
The multiple small spikes are due to the execution of \RankUpdate.
To update the number of latent dynamics, \RankUpdate runs \ModelEstimation multiple times, which increases the computation time at that time step.
The computation time of DISMO depends on the number of factors.
Consequently, as the data stream length increases, the number of factors increases, leading to a gradual rise in computation time.
FluxCube demands substantial computational time because of the repeated training needed to estimate optimal location groups.
The right column in \autoref{fig:scalability} presents the average computation time for the entire tensor streams.
\propose achieves the fastest computation time.


\begin{figure}[h]
    \centering
    \begin{tabular}{cc}
        \begin{minipage}[b]{0.45\linewidth}
        \centering
        \includegraphics[width=\linewidth]{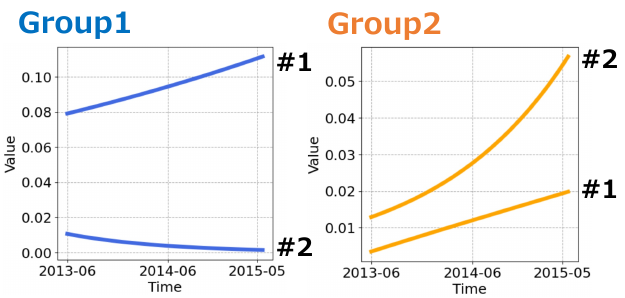}
        \includegraphics[width=\linewidth]{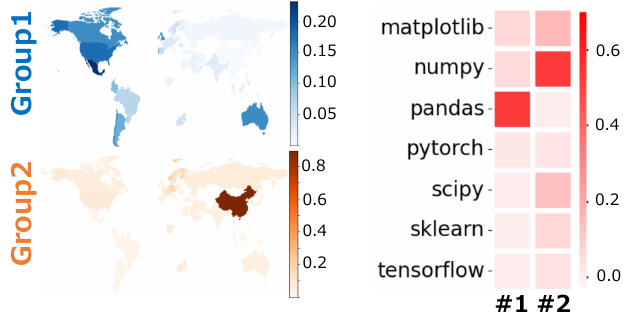}
        \includegraphics[width=\linewidth]{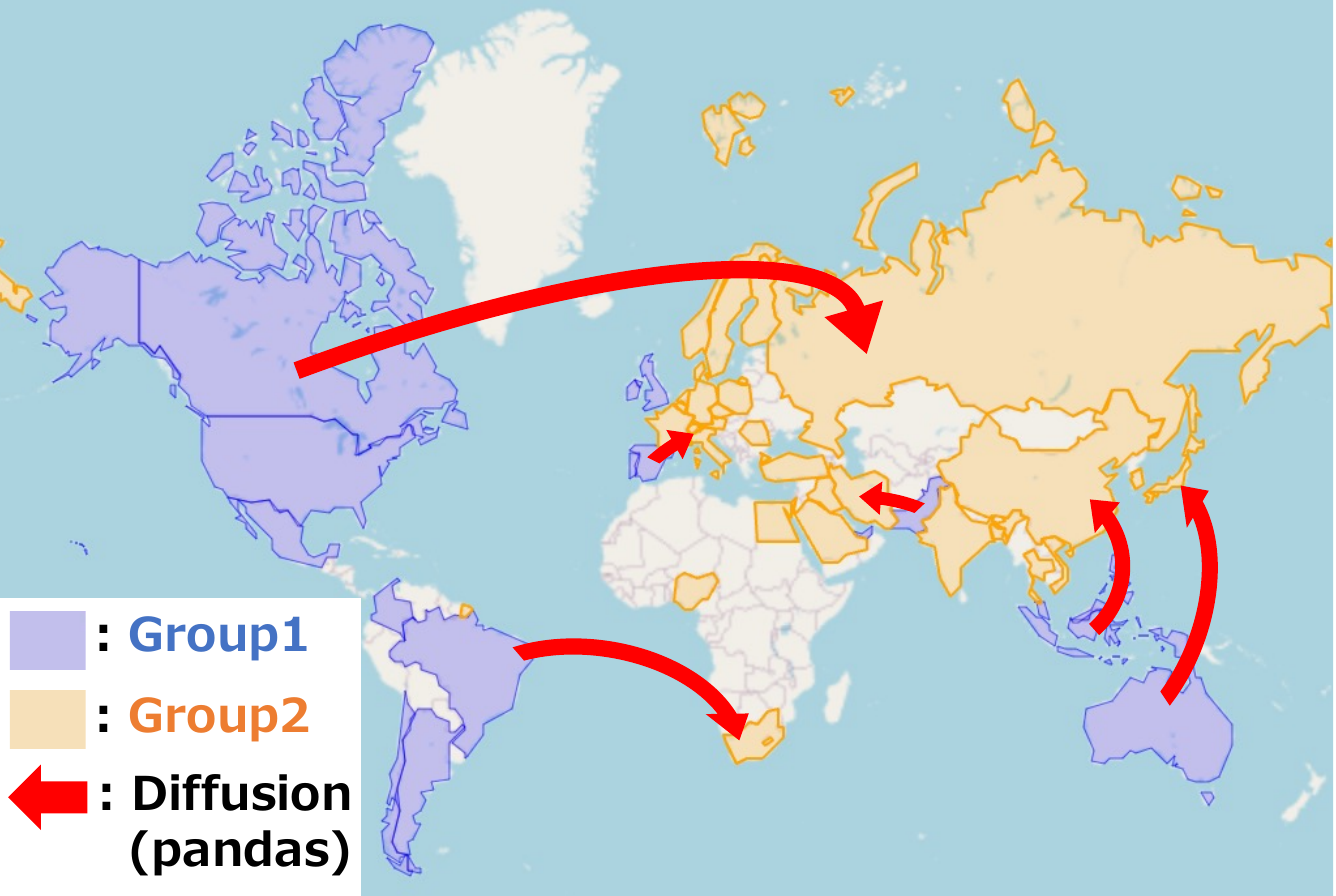}
        (a) 2013-2015
        \end{minipage}
         &  
         \begin{minipage}[b]{0.46\linewidth}
         \centering
         \includegraphics[width=\linewidth]{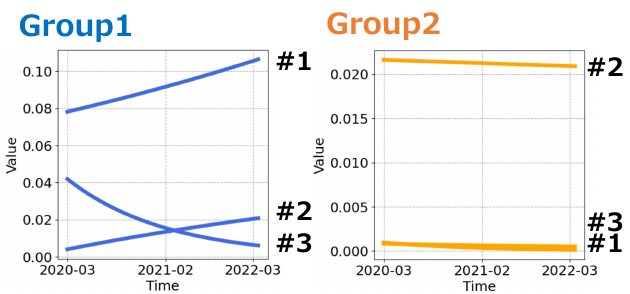}
         \includegraphics[width=\linewidth]{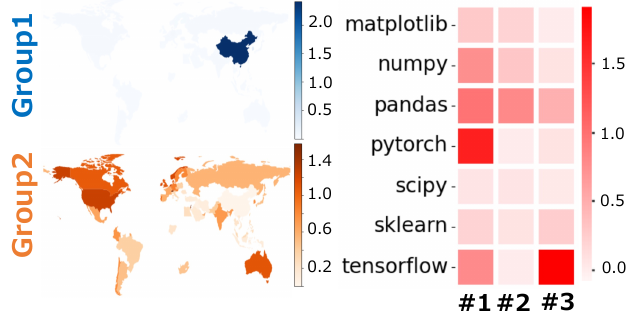}
         \includegraphics[width=\linewidth]{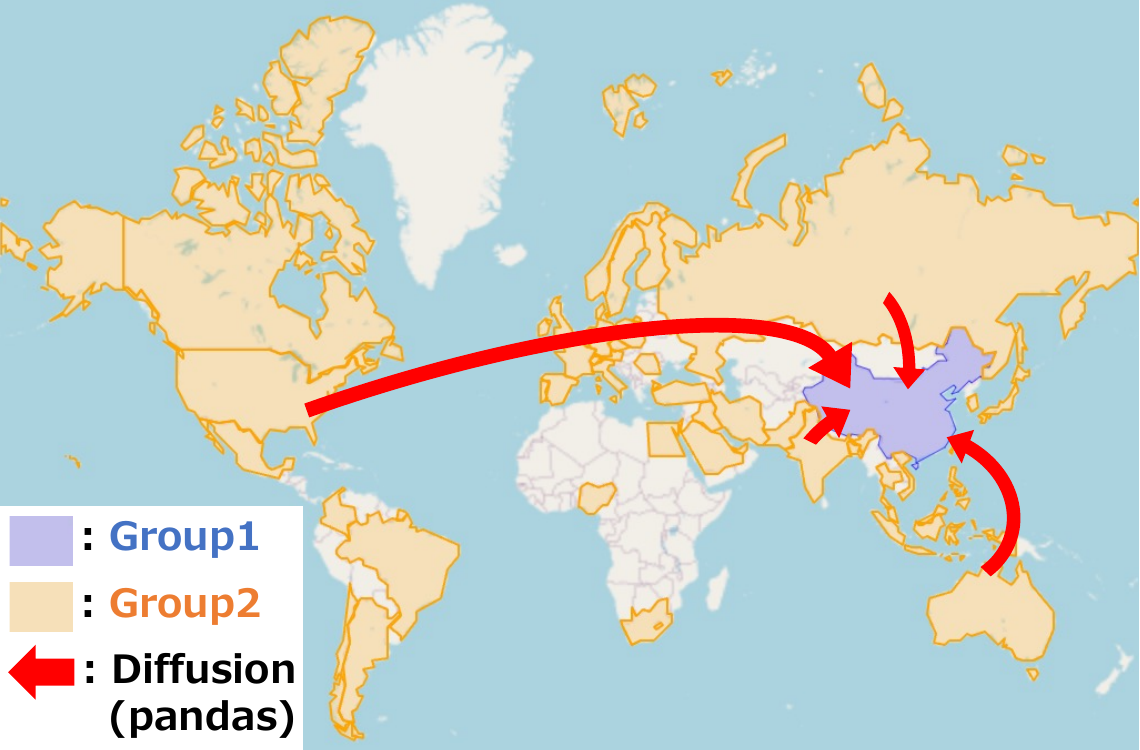}
        (b) 2020-2022
         \end{minipage}
    \end{tabular}
    \caption{Modeling results for the \pythonlib data over two periods ((a) Jun 2013 - May 2015; (b) Mar 2020 - Feb 2022). Our model decomposes the input tensor into the latent dynamics $\Wcore$ (top), the location factor $\Wloc$, and the keyword factor $\Wkey$ (middle). Also, our reaction-diffusion system captures the interest diffusion of ``pandas'' (bottom). 
    }
    \label{fig:exp_python}
\end{figure}

\mypara{Q3: Effectiveness}
Finally, we describe how effectively \propose captures trends and diffusion patterns.
The result for the \device dataset was presented in Section \ref{sec:intro}.

\autoref{fig:exp_python} shows the modeling results for the \pythonlib dataset over two periods.
The top row corresponds to the latent dynamics $\Wcore$ and the middle row corresponds to the location/keyword factors $\Wloc$ and $\Wkey$.
Again, note that the location/keyword factors are non-negative matrices, indicating the degree to which each dynamics is related to each country/keyword.
The bottom row shows the diffusion patterns captured by the reaction-diffusion system.
As for the location groups, a group with a higher $\Wloc$ value is assigned to each country.
Here, the two models in the full parameter set $\F$ are picked up and shown, corresponding to the periods (a) 2013-2015 and (b) 2020-2022 respectively.
As seen in \autoref{fig:exp_python} (a) and (b), the proposed method captures time-varying patterns thanks to the model switching algorithm, and the rank update algorithm (Algorithm \ref{alg:RankUpdate}) changes the dimension of the latent dynamics in each period (i.e. the value of $d_k$ changes from 2 to 3).
In 2013-2015 (i.e., \autoref{fig:exp_python} (a)), the diffusion from Group 1 to 2 was captured for Dynamics \#1.
$\Wkey$ (\#1) shows that Dynamics \#1 is strongly related to ``pandas'', which can be interpreted as a diffusion of interest in ``pandas'' from countries in Group 1 (e.g. North America and Australia) to countries in Group 2 (e.g. Asian countries, mainly China).
Also, in 2020-2022 (i.e., \autoref{fig:exp_python} (b)), the location factor (Group 1) and the keyword factor (\#1) indicate that Dynamics \#1 of Group 1, which shows an increasing trend, is strongly related to the ``numpy,'' ``pandas,'' ``pytorch,'' and China.
Furthermore, Dynamics \#3 of Group 1, which has a decreasing trend, is strongly related to ``tensorflow'' and China.
These results can be interpreted as a growing interest in machine learning libraries in China and a transition of the main tool from ``tensorflow'' to ``pytorch''.
Our reaction-diffusion system also captured a diffusion pattern from Group 2 to Group 1 in Dynamics \#2.
We can read that interest in ``matplotlib,'' ``numpy,'' and ``pandas'', which are related to Dynamics \#2, is diffusing to China.

\section{Conclusion}
    \label{sec:conclusion}
    
In this paper, we proposed an effective modeling and forecasting method, namely \propose , for social activity data streams.
Using the reaction-diffusion system, our model captures latent dynamics such as trends and diffusion patterns in the input tensor.
Also, our algorithm adapts our model to changing patterns in the data stream and allows continuous modeling.
Our approach has the following properties.
(a) \textit{Interpretable}: Our model decomposes the original data into latent dynamics consisting of trends and interest diffusion patterns and location/keyword factors, which improves the interpretability for high-dimensional original data.
(b) \textit{Automatic}: \propose does not require any parameter tuning.
(c) \textit{Scalable}: the computation time of \propose is independent of the time series length.
An experimental evaluation using real datasets obtained from GoogleTrends and COVID-19 Open Data Repository showed that our proposed method achieved higher forecasting accuracy in less computation time than its competitors.

\small
\section*{acknowledgement}
This work was partly supported by
JST BOOST, Japan Grant Number JPMJBS2402, 
“Program for Leading Graduate Schools” of the Osaka University, Japan, 
JST CREST JPMJCR23M3,
JSPS KAKENHI Grant-in-Aid for Scientific Research Number
JP22K17896,    
JP23K16889,    
Research Institute of Science and Technology for Society, Japan, Grant Number JPMJRS23L4. 

\normalfont

\bibliographystyle{ACM-Reference-Format}
\balance
\bibliography{%
BIB/forecast,
BIB/stream,
BIB/web,
BIB/physics,
BIB/mdl,
BIB/tensor
}

\appendix
\section*{Appendix}
\section{Algorithm}
\label{appendix:algo}

In this section, we describe our proposed algorithm in detail.

\subsection{Initialization}
\label{appendix:init}
Here, we detail the initialization step of Algorithm \ref{alg:ModelEstimation}.
To decompose the current window $\Xc$ into the trend tensor $\Xd$ and the seasonal tensor $\Xs$, we propose using STL \cite{STL}, which is a robust and general algorithm for seasonal-trend decomposition.
It decomposes time series data into the sum of the trend, the seasonal pattern, and the residual term.
Furthermore, $\Stime$, $\Skey$, and $\Sloc$ are obtained by performing PARAFAC decomposition on $\Xs$.
For the initialization of $\Wkey$ and $\Wloc$, we use Nonnegative Tucker Decomposition (NTD) \cite{NTD}.
We apply NTD to $\Xd$ and use the obtained keyword and location factors as the initial values of $\Wkey$ and $\Wloc$, respectively.

\subsection{ModelEstimation}
\label{appendix:ModelEstimation}
Algorithm \ref{alg:ModelEstimation} shows the \ModelEstimation in detail.
It is based on alternating least squares (ALS).
Specifically, we alternate between updating the reaction-diffusion system and updating each factor for the trend tensor, each factor for the seasonal tensor, and the outlier tensor, while keeping the other parameters fixed.
\begin{algorithm}[h]
    \footnotesize
    \caption{\ModelEstimation ($\Xc ,d_{k}, d_{l}, d_{s}$)}
    \label{alg:ModelEstimation}
\begin{algorithmic}[1]
    \REQUIRE
        Current tensor $\Xc$ and
        current ranks $d_{k}, d_{l}, d_{s}$
    \ENSURE
        Model parameter $\Theta$ \\
\STATE $\{ \theta_d ,\theta_s \} = \textrm{Initialization}(\Xc ,d_{k} ,d_{l} ,d_{s})$; $\Xo = 0$;
\REPEAT
    \STATE /* Update reaction-diffusion system with Equation(\ref{eq:lmfit}) */
    \STATE $\theta = \argmin_{\theta^{*}} ||\Xc - \Xs - \Xo - f_{\theta_d}(\{ \theta^{*},\Wkey ,\Wloc \} ,w_0,L)||$; 
    \STATE /* Update factors with Equation(\ref{eq:factorupdate}) */
    \STATE $\Wcore = f_{\theta}(\theta ,w_0,L)$;
    \STATE $\{\Wkey ,\Wloc\} = \rm{Update}(\Wcore ,\Wloc ,\Wkey ,\Xc - \Xs - \Xo)$; 
    \STATE /* Update seasonal factors */
    \STATE $\{\Stime ,\Skey ,\Sloc\} = \rm{Update}(\Stime ,\Skey ,\Sloc ,\Xc - \Xd - \Xo)$;
    \STATE /* Update outlier tensor */
    \STATE $\Xo = \rm{Sparsify}(\Xc - \Xd - \Xs)$;
\UNTIL{convergence;}
\STATE {\bf return} $\Theta = \{ \theta_d ,\theta_s ,\Xo \}$;
\end{algorithmic}
\end{algorithm}

\subsection{Theoretical analysis}
\label{appendix:proof}
We provide the proofs of Lemma \autoref{lemma1}.

\begin{proof}
The number of parameters to be estimated for our reaction-diffusion system is $d_{k}d_{l}^{2}$.
Hence, the time complexity of parameter optimization according to \autoref{eq:lmfit} is $O(d_{k}d_{l}^{2})$.
Also, the time complexity of updating $\Wkey \in \mathbb{R}^{d_k \times k}$ and $\Wloc \in \mathbb{R}^{d_l \times l}$ according to \autoref{eq:factorupdate} is $O(kd_{k} + ld_{l})$.
Hence, the computational complexity of Algorithm \ref{alg:ModelEstimation} is $O(\# iter \cdot (d_{k}d_{l}^{2} + kd_{k} + ld_{l}))$.
Since $\# iter$ is a negligibly small constant value, the total time complexity is $O(d_{k}d_{l}^{2} + kd_{k} + ld_{l})$.
\end{proof}

\subsection{ModelUpdate}
\label{appendix:modelupdate}
Here, we describe \ModelUpdate in detail.
Algorithm \ref{alg:ModelUpdate} shows the overall procedure of \ModelUpdate .
As described in Section \ref{sec:algo}, given the current tensor $\Xc$, the full parameter set $\F$, and the candidate model parameter $\Theta$, the algorithm compares the total cost of adding $\Theta$ to $\F$ with the total cost without the addition. 
If the total cost is lower with the new model, we add the candidate model parameter to the full parameter set and model the current tensor using the new model.

After switching the model, \ModelUpdate updates the ranks (i.e., the number of latent dynamics $d_{k}, d_{l}$ and $d_{s}$) by running \RankUpdate algorithm.
Algorithm \ref{alg:RankUpdate} shows the \RankUpdate in detail.
We update these values based on the total cost to better represent the current tensor.
We assume that the optimal values change gradually in streaming scenarios and only search for cases where only one of these values changes by one.

\begin{algorithm}[h]
    \footnotesize
    \caption{\ModelUpdate ($\Xc ,\F ,\Theta' , d_{k}, d_{l}, d_{s}$)}
    \label{alg:ModelUpdate}
\begin{algorithmic}[1]
    \REQUIRE
        (a) Current tensor $\Xc$, (b) Full parameter set $\F$; \\
        \hspace{1.5em}
        (c) Candidate parameter $\Theta$, (d) Current ranks $d_{k}, d_{l}, d_{s}$ \\
    \ENSURE
        (a) Full parameter set $\F$ and 
        (b) updated ranks $d_{k}, d_{l}, d_{s}$ \\
\STATE $\F' \leftarrow \F$
\IF{$\costT{\Xc}{\F' \cup \Theta'}$ is less than $\costT{\Xc}{\F'}$}
\STATE $\F' \leftarrow \F' \cup \Theta'$;
\STATE $C = \costT{\Xc}{\F' \cup \Theta'}$;
\STATE $\{d_{k}, d_{l}, d_{s}\} =$ \RankUpdate $(\Xc ,d_{k}, d_{l}, d_{s}, C)$;
\ENDIF
\STATE {\bf return} $\F' ,d_{k}, d_{l}, d_{s}$;
\end{algorithmic}
\end{algorithm}
\begin{algorithm}[h]
    \footnotesize
    \caption{RankUpdate ($\Xc , d_{k}, d_{l}, d_{s}$)}
    \label{alg:RankUpdate}
\begin{algorithmic}[1]
    \REQUIRE
        Current tensor $\Xc$ 
        , current ranks $d_{k}, d_{l}, d_{s}$, and current cost $C$ \\
    \ENSURE
        Updated ranks $d_{k}, d_{l}, d_{s}$ \\

\STATE /* Update $d_{k}, d_{l}, d_{s}$ */
\STATE $C_{best} = C$ 
\FOR{$(d_{k}', d_{l}', d_{s}')$ \textbf{in} $(d_{k}-1, d_{l}, d_{s}),(d_{k}, d_{l}+1, d_{s}),(d_{k}, d_{l}-1, d_{s}),(d_{k}, d_{l}+1, d_{s}),(d_{k}, d_{l}, d_{s}-1),(d_{k}, d_{l}, d_{s}+1)$}
    \STATE $\Theta' =$ \ModelEstimation $(\Xc, d_{k}', d_{l}', d_{s}')$
    \IF{$C_{best} > \costT{\Xc}{\Theta'}$}
        \STATE $C_{best} = \costT{\Xc}{\Theta'}$
        \STATE $\{ d_{k}, d_{l}, d_{s}\} = \{ d_{k}', d_{l}', d_{s}'\} $
    \ENDIF
\ENDFOR
\STATE {\bf return} $d_{k}, d_{l}, d_{s}$;
\end{algorithmic}
\end{algorithm}

\section{Experimental evaluation}
\label{appendix:exp}

\subsection{Experimental Settings}
\label{appendix:exp_setting}
Table \ref{tab:dataset} shows the details of the datasets we use.
We compare our algorithm with the following state-of-the-art models for time series forecasting.
\begin{itemize}
    \item DISMO \cite{dismo}: a stream algorithm that is designed to discover dynamic interactions and seasonality in a multi-order tensor stream. It has no parameters to set.
    \item FluxCube \cite{fluxcube}: a method for modeling diffusion patterns in temporal tensors based on physics-informed neural networks.
    We set $\{ 16, 32, 64 \}$ as the size of the RNN hidden layers, and $\{ 0.1, 0.2, 0.3 \}$ as the weight of the loss function.
    \item PatchTST \cite{PatchTST}: a stete-of-the-art Transformer based model, which incorporates patching and channel independence.
    We set $\{ 2,4,8,16,24 \}$ as the patch length, $\{ 2,4,8,16,24 \}$ as the stride, $\{ 4,8,16 \}$ as the number of heads, $\{ 0.001, 0.0001 \}$ as the learning rate, and $\{ 0, 0.1, 0.2, 0.3 \}$ as the dropout rate.
    \item Autoformer \cite{Autoformer}: a Transformer based model, which includes an Auto-Correlation mechanism that discovers period-based dependencies. 
    We set $\{ 2,4,8,16 \}$ as the number of heads, $\{ 1,2,3 \}$ as the number of layers of the encoder, $\{ 1,2,3 \}$ as the number of layers of the decoder, $\{ 1,2,3 \}$ as the number of factors, and $\{ 0.001, 0.0001 \}$ as the learning rate.
    \item DeepAR \cite{DeepAR}: a probabilistic forecasting model based on an autoregressive RNN. 
    We set $\{ 2,3,4 \}$ as the number of layers, $\{ 32,64,128,256 \}$ as the size of the hidden layers of RNN, $\{ 0.001, 0.005, 0.0001 \}$ as the learning rate, and $\{ 0,0.1,0.2,0.3 \}$ as the dropout rate.
    \item LaST \cite{LaST}: a representation learning based method that infers a couple of representations depicting trends and seasonality of time series.
    We set $\{ 16,32,64,128 \}$ as the dimension of latent representations, $\{ 0.1,0.2,0.3 \}$ as the dropout rate, and 0.001 as the initial learning rate.
    \item CoST \cite{CoST}: a contrastive learning based method that infers a couple of representations depicting trends and seasonality of time series.
    We set $\{ 160, 320, 640 \}$ as the representation dimensions, and $\{ 0.0005, 0.005, 0.05 \}$ as the weight for loss function.
\end{itemize}

We trained PatchTST, Autoformer, and DeepAR for 100 epochs with early stopping, FluxCube for 2000 epochs with early stopping, LaST for 1000 epochs with early stopping, and Cost for 600 epochs.
We also trained the model every year for GoogleTrends datasets and every month for the Covid-19 dataset.
In order to prevent underestimation of the baselines, we ran PatchTST, Autoformer, and DeepAR for three different lengths of the input tensor, denoted as $\lc \in \{ 52, 104, 156 \}$, and select the best results to establish strong and robust baselines for GoogleTrends datasets.
For COVID-19 dataset, we set $56$ as the length of the input tensor for all the baselines except FluxCube.
We conducted all our experiments on an Intel Xeon Platinum 8268 2.9GHz CPU and RTX A6000 GPU with 512GB of memory.

\mypara{Hyperparameter tuning}
For GoogleTrends datasets, we used the first 5-year data (i.e., 2010-2014) for training, and data in 2015 as validation data for hyperparameter tuning.
Also, for COVID-19 dataset, we used data in 2020 for training, and data in January 2021 as validation data.
Hyperparameters were tuned using Optuna or grid search.
In the initialization step of \propose (i.e., Algorithm \ref{alg:DTracker}), we searched for the optimal numbers of latent dynamics by employing grid search in $(d_{k},d_{l}) \in [2,4]$ and $d_{s} \in [0,4]$, so that Equation (\ref{eq:totalcost}) was minimized.

\begin{table}[t]
    \centering
    \caption{Dataset description}
    \scalebox{0.85}{
    \begin{tabular}{c|l|c|c}
        \toprule
        ID & Name & Keywords & Data size \\ 
        \midrule
        \# 1 & \device    & Apple Watch/Fitbit/Galaxy/Google Pixel/ & $(676, 8, 50)$ \\
             &            & Macbook/Xperia/iPad/iPhone & \\ 
        \midrule
        \# 2 & \pythonlib & matplotlib/numpy/pandas/pytorch/ & $(676,7,50)$ \\
             &            & scipy/sklearn/tensorflow \\
        \midrule
        \# 3 & \vod       & Apple TV/DAZN/HBO/Prime Video/hulu/ & $(676,7,50)$\\
             &            & netflix/youtube & \\
        \midrule
        \# 4 & \chatapp   & facebook/instagram/slack/snapchat/ & $(676,8,50)$\\
             &            & tiktok/twitter/viber/whatsapp \\
        \midrule
        \# 5 & \programming & CSS/HTML/Java/JavaScript/PHP/ & $(676,9,50)$\\
             &              & Ruby/SQL/python/swift \\
        \midrule
        \multicolumn{3}{l}{Infection dataset} \\
        \midrule
        \# 6 & \covid       & new infected /new deceased & $(990,2,50)$\\
        \bottomrule
    \end{tabular}
    }
    \label{tab:dataset}
\end{table}

\begin{figure}[h]
  \begin{tabular}{ll}
    \begin{minipage}[b]{0.47\linewidth}
        \centering
        \includegraphics[width=\linewidth]{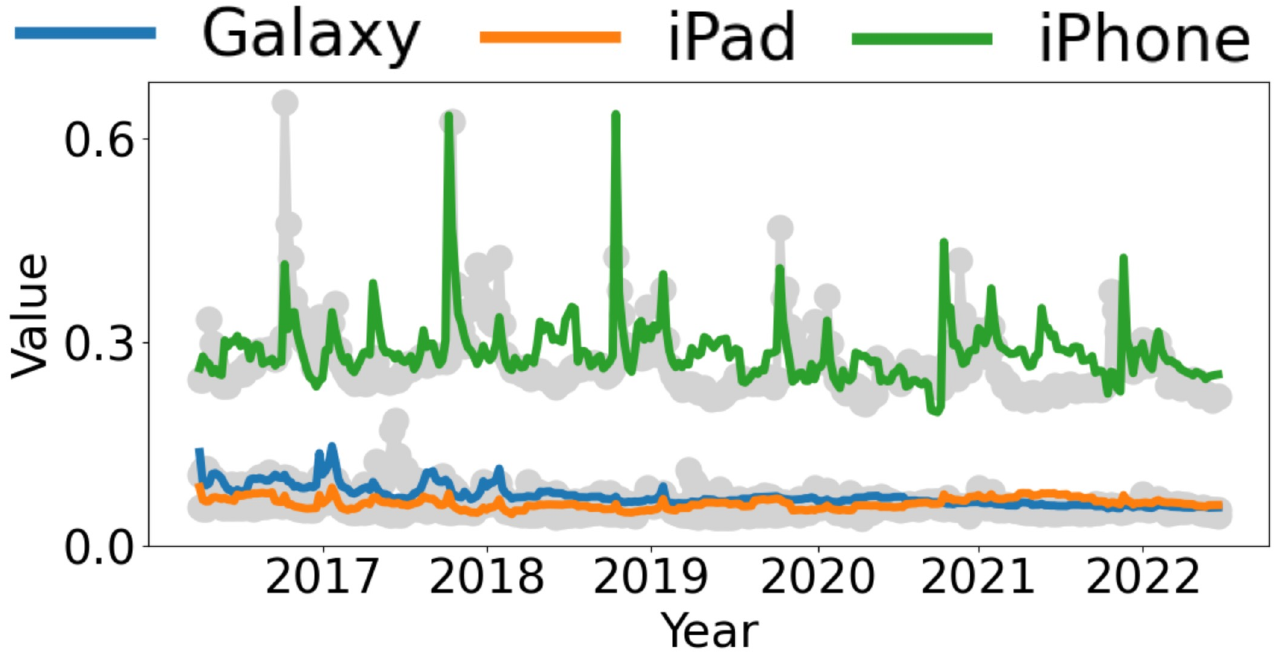} \\
        (a) The three keywords in US
    \end{minipage}
    &
    \begin{minipage}[b]{0.47\linewidth}
        \centering
        \includegraphics[width=\linewidth]{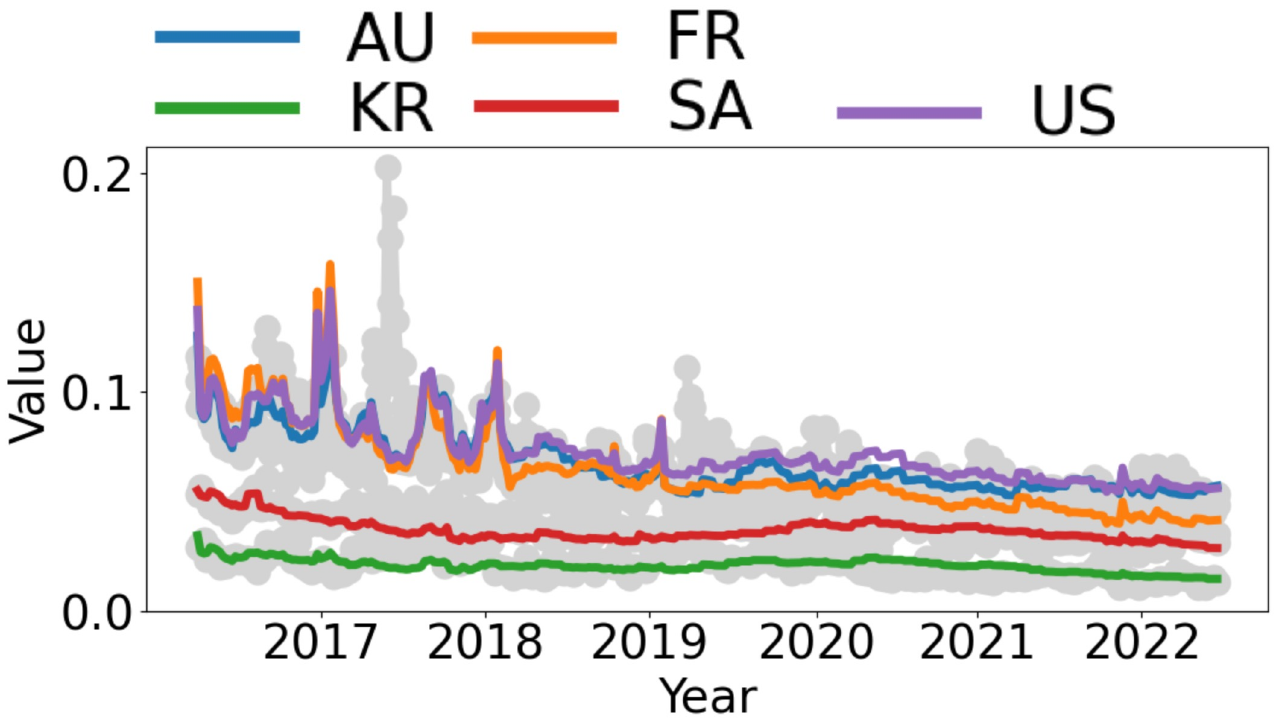} \\
        (b) ``Galaxy" in five countries
    \end{minipage}
  \end{tabular}
  \caption{Forecasting power of \propose : Our model accurately captures trends and seasonal patterns that vary by location and keyword, and performs stream forecasting. }
  \label{fig:forecasting}
\end{figure}

\subsection{Effectiveness}
\label{appendix:effectiveness}
\autoref{fig:exp_chatapp} shows the modeling results for the \chatapp dataset between 2016 and 2018.
Dynamics \# 1 in Group 2 shows an increasing trend.
The keyword factor (\# 1) shows that Dynamics \# 1 is strongly related to ``facebook'' and the location factor (Group 2) shows that Group 2 is related to North America, Australia, and European countries.
These results indicate that interest in ``facebook'' has a pattern of growing interest in these countries.
Also, Dynamics \# 2 in Group 2 has a slight upward trend and is strongly related to ``instagram'', North America and European countries (See the \# 2 of the keyword factor and the Group 2 of the location factor).
This indicates a growing interest in ``instagram'' in these countries.
In addition, in 2016-2018, our reaction-diffusion system showed that for Dynamics \# 1 there is a pattern of diffusion from Group 1 to Group 3.
This indicates that interest in ``facebook'', which is related to Dynamics $\# 1$, diffused from countries strongly related to Group 1 to those strongly related to Group 2. 
\autoref{fig:exp_chatapp} (c-ii) show the diffusion pattern on a map.
Furthermore, our algorithm can adapt to dynamic changes in patterns by incrementally updating the model.
\autoref{fig:exp_chatapp} (c-i) illustrates the diffusion pattern captured in 2010-2012.
It can be interpreted as a change in the flow of diffusion of interest in ``facebook'' over time.

In addition, \autoref{fig:forecasting} shows forecasting results of \propose for the \device dataset.
Our method models only the most recent tensor and continuously generates future values.
\autoref{fig:forecasting} (a) shows forecasting results for three devices in the United States.
Our method accurately captures the seasonality for the ``iPhone" (due to the launch of new models) and predicts future seasonal events.
\autoref{fig:forecasting} (b) shows forecasting results for ``Galaxy" in the five countries.
Our method accurately captures the declining trend, which varies in strength from country to country, and performs accurate forecasting.

\begin{figure}[h]
    \centering
    \includegraphics[width=\linewidth]{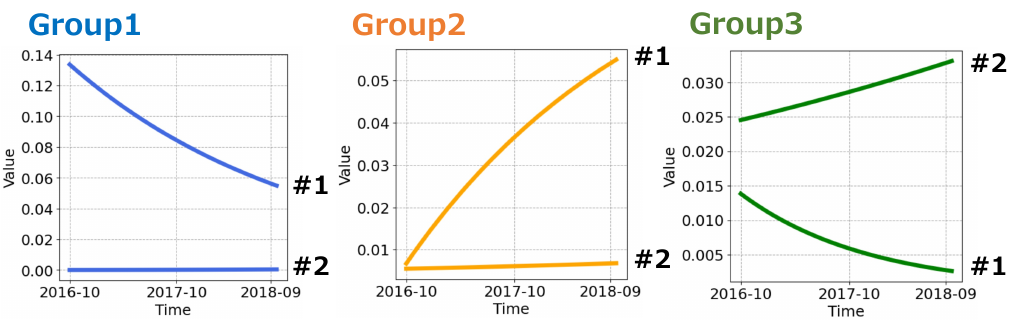}
    (a) Latent dynamics generated from reaction-diffusion system
    \begin{tabular}{cc}
        \begin{minipage}[b]{0.65\linewidth}
        \centering
        \includegraphics[width=\linewidth]{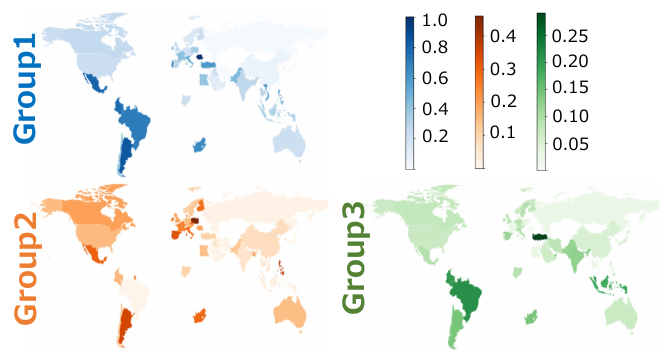}
        \end{minipage}
         &  
         \begin{minipage}[b]{0.25\linewidth}
         \centering
         \includegraphics[width=\linewidth]{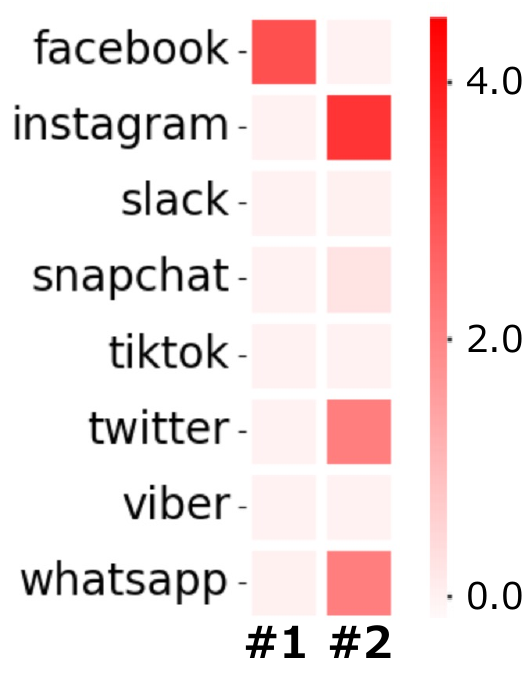}
         \end{minipage}
    \end{tabular}
    (b) Location and keyword factors
    \begin{tabular}{cc}
        \begin{minipage}[b]{0.45\linewidth}
        \centering
        \includegraphics[width=\linewidth]{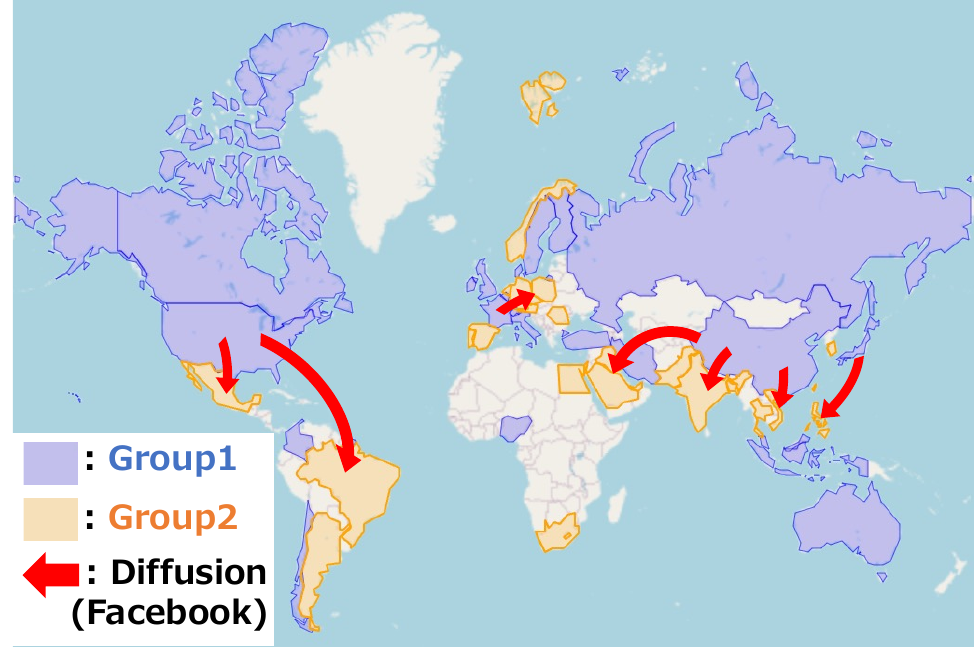}
        (c-i) 2010-2012
        \end{minipage}
         &  
         \begin{minipage}[b]{0.45\linewidth}
         \centering
         \includegraphics[width=\linewidth]{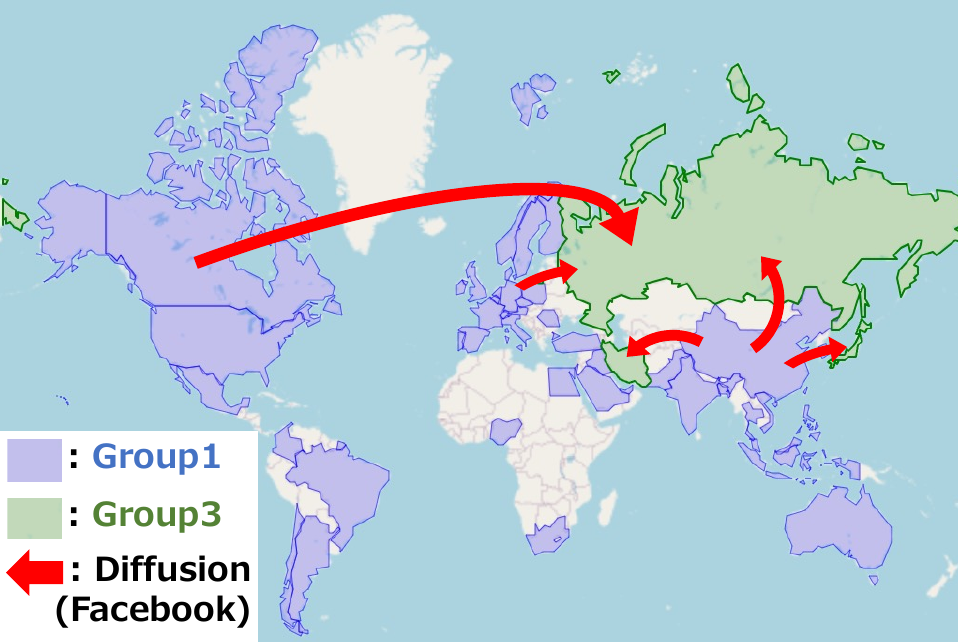}
        (c-ii) 2016-2018
         \end{minipage}
    \end{tabular}
    (c) DIffusion patterns over the two periods
    \caption{Modeling results for the \chatapp data. 
    }
    \label{fig:exp_chatapp}
\end{figure}

\clearpage


\end{document}